\documentclass{sig-alternate-05-2015}
\DeclareMathSizes{9}{8.5}{7.5}{7}
\usepackage{verbatim,multirow,mathtools,enumitem,textcomp,capt-of,url,color,lipsum,epstopdf,graphicx,relsize,balance,times,layouts,amssymb}
\usepackage[ruled,vlined,linesnumbered,nofillcomment]{algorithm2e}
\SetArgSty{textrm}
\SetCommentSty{<textrm>}
\SetAlFnt{\small}
\SetVlineSkip{1pt}
\SetInd{0.5em}{0.7em}
\usepackage[skip=1pt,font=bf]{caption}
\usepackage[caption=false]{subfig}
\setlist{leftmargin=1.2em, itemsep=2pt, topsep=2pt,parsep=1pt}

\newtheorem{definition}{Definition}[section]
\newtheorem{proposition}[definition]{Proposition}

\newtheorem{theorem}[definition]{Theorem}
\newtheorem{example}[definition]{Example}

\DeclareMathOperator*{\avg}{Avg}

\newcommand{\eat}[1]{}
\newcommand{\rbox}{\hfill $\Box$}
\newcommand{\ie}{{\em i.e.}}
\newcommand{\eg}{{\em e.g.}}

\newcommand{\observe}{{\cal O}}

\newcommand{\model}{\textsc{Hybrid}}  

\newcommand{\fr}[1]{{\color{red} [furong: #1]}}

\newcommand\Mark[1]{\textsuperscript#1}

\begin{document}

\doi{}

\isbn{}

	\setlength{\abovedisplayskip}{0pt}
	\setlength{\belowdisplayskip}{0pt}
	\setlength{\textfloatsep}{14pt minus 4 pt} 
	\setlength{\intextsep}{8pt minus 2 pt} 


\title{Discovering Multiple Truths with a Hybrid Model}

\author{
\alignauthor
Furong Li\Mark{\dag} ~~~~~~~~~~~ Xin Luna Dong\Mark{\ddag} ~~~~~~~~~~~ Anno Langen\Mark{\ddag} ~~~~~~~~~~~ Yang Li\Mark{\ddag}\\
\affaddr{\Mark{\dag}National University of Singapore~~~~~~~~~~~~~~~~~~~~~~~~~~~~~~~~~~\Mark{\ddag}Google Inc., Mountain View, CA, USA}  \\ 
\email{\hspace{2mm}furongli@comp.nus.edu.sg~~~~~~~~~~~~~~~~~~~~~~~~~~~\{lunadong, arl, ngli\}@google.com}
}

\maketitle

\begin{abstract}
Many data management applications require integrating information from multiple sources.
The sources may not be accurate and provide erroneous values.
We thus have to identify the true values from conflicting observations made by the sources.
The problem is further complicated when there may exist multiple truths (\eg, a book written by several authors).
In this paper we propose a model called \model{} that jointly makes two decisions: {\em how many truths there are}, and {\em what they are}. 
It considers the conflicts between values as important evidence for ruling out wrong values, while keeps the flexibility of allowing multiple truths.
In this way, \model{} is able to achieve both high precision and high recall.
\end{abstract}

\section{Introduction}
\label{multi:sec:intro}

When consolidating information from different sources, we may observe different values provided for the same entity.
Consequently, we need to identify the correct values from conflicting observations made by multiple sources, which is known as the {\em data fusion} (or {\em truth discovery}) problem~\cite{fusionsurvey,fusionsurvey15}.
We illustrate the problem using the example below.

\begin{example}
	Table~\ref{multi:tbl:triples} shows the information collected from three sources regarding equipments of two winter sports: {\em ice hockey} and {\em snowboarding}.
	We can see that four different values are provided for the entity {\em ice hockey} ({\em helmet, stick, boots} and {\em skis}), while only the first two are correct.
	The goal of the truth discovery problem is to identify the correct values from Table~\ref{multi:tbl:triples}.
	\label{multi:exa:problem}
\end{example}

The simplest solution to the truth discovery problem is {\em majority vote}: consider the value provided by the largest number of sources as the truth.
For example, for the entity {\em snowboarding}, we select {\em board} as the truth since it is provided by two sources while {\em neck guard} is only provided by one source.
However oftentimes different sources have different qualities, and one may want to distinguish them.
The authors of \cite{truthFinder,accu} measure the quality of a source $s$ by its {\em accuracy}, which is the probability that a value provided by $s$ is correct.
Then the truth can be decided through a {\em weighted vote}, where a source with higher accuracy is assigned to a higher weight;
the value with the highest vote is selected as the truth.
The intuition behind this approach is that values provided by more accurate sources are more likely to be true.

\begin{table}[h]
	\centering
	\caption{Information collected from different sources regarding equipments of various winter sports.
		$\surd/\times$ indicates the correctness.}
	\begin{tabular}{@{} r l l l l} \cline{2-5}
		& entity & attribute & value & sources \\ \cline{2-5}
		$\surd~ o_1$ & ice hockey & equipments & helmet & $s_1, s_3$\\ 
		$\surd~ o_2$ & ice hockey & equipments & stick & $s_1, s_2$\\  
		$\times~ o_3$ & ice hockey & equipments & boots & $s_2$\\ 
		$\times~ o_4$ & ice hockey & equipments & skis & $s_3$\\ \cline{2-5}
		$\surd~ o_5$ & snowboarding & equipments & board & $s_2, s_3$\\ 
		$\times~ o_6$ & snowboarding & equipments & neck guard & $s_1$\\ \cline{2-5}
	\end{tabular} 
	\label{multi:tbl:triples}
\end{table}

The limitation of the above methods is that when multiple truths exist, 
they at best find one of them while missing the rest. 
We thus refer to them as {\em single-truth} approaches.
While truth discovery algorithms usually compute a probability $p(v)$ of each value $v$ being true, in single-truth approaches, the probabilities of all values sum up to 1 since they assume there is only one true value.

\begin{example}
	We use {\sc Accu}~\cite{accu} as a representative of single-truth approaches, and then compute the probabilities of the values provided from ice hockey equipments.
	Assuming all sources have the same accuracy 0.6, we obtain the probabilities in Table~\ref{multi:tbl:prob_example} (see the first line).
	We observe that the probabilities of the four values add up to 1, so even the true values ({\em helmet} and {\em stick}) have rather low probabilities.
	\label{multi:exa:single}
\end{example}

\begin{table}
	\centering
	\caption{Value probabilities computed by different approaches for {\em (ice hockey, equipments)}.}
	\begin{tabular}{@{}lrrrr} \hline
		& helmet & stick & boots & skis \\ \cline{2-5}
		Single-truth~\cite{accu} & 0.47 & 0.47 & 0.03 & 0.03 \\ 
		Multi-truth~\cite{precrec} & 0.63 & 0.63 & 0.54 & 0.54 \\ 
		{\sc Hybrid} & 0.92 & 0.92 & 0.08 & 0.08 \\ \hline
	\end{tabular} 
	\label{multi:tbl:prob_example}
\end{table}

To address the above problem, {\em multi-truth} approaches~\cite{LTM,precrec} have been proposed recently.
They compute the probability for each value separately, and thus do not require the probabilities of all values sum up to 1.
Instead, they compute both the probability $p(v)$ of $v$ being true, and the probability $p(\neg v)$ of $v$ being false, where $p(v)+p(\neg v)=1$.
Then a value $v$ is consider true if $p(v) > p(\neg v)$, that is, $p(v) > 0.5$.

An {\em unknown semantics} is used to capture the nature of multi-truth: if a source $s$ does not provide the value $v$,
$s$ means that it {\em does not know} whether or not $v$ is correct (instead of saying $v$ is {\em incorrect}).
Thus, apart from {\em accuracy} (also called {\em precision} in some methods), 
multi-truth approaches also measure the quality of a source $s$
by its {\em recall}, the probability that a truth is provided by $s$.
Intuitively, values provided by high-precision sources are likely to be true (\ie, a higher $p(v)$), 
and values not provided by high-recall sources are likely to be false (\ie, a higher $p(\neg v)$).
In this way, they derive $p(v)$ and $p(\neg v)$ from the precision and recall of the sources, and then normalize with the equation $p(v)+p(\neg v)=1$.

\begin{example}
	Table~\ref{multi:tbl:prob_example} also shows the value probabilities computed by the multi-truth approach {\sc PrecRec}~\cite{precrec}.
	Assuming a precision of 0.6 and recall of 0.5 for each source, {\sc PrecRec} will decide
	that all provided values are true, resulting in false positives (\ie, {\em boots} and {\em board}).
	\label{multi:exa:multi}
\end{example}

In practice, even multi-truth items often have only a few true values instead of infinite number of truths. 
Existing methods~\cite{LTM,precrec} cannot capture this because they decide the truthfulness of each value independently, 
without considering other values provided for the entity and thus lack a global view of the entity. 
As a result, they suffer from low precision when the sources have low coverage or noisy observations (as shown later in Section~\ref{multi:sec:exp}).

In this paper we introduce a new solution to the truth discovery problem, called \model{}, which works for multi-truth applications.
Based on the values provided for an entity, {\sc Hybrid} makes two decisions: (i) how many truths there are, and
(ii) what they are. 
Essentially it interleaves the two decisions and finds the truths one by one.
Conditioning on a sequence of true values that have been selected previously, it computes (1) the probability of a value $v$ being the next truth, and (2) the probability that there is no more truth.
In this way, {\sc Hybrid} combines the flexibility of the multi-truth approaches of allowing multiple truths for an entity, 
and the inherent strength of the single-truth approaches of considering conflicts between values
as important evidence for ruling out wrong values.
Therefore, it obtains both high precision and high recall.

Note that the {\em multi-truth} setting should be considered more general than the {\em single-truth} setting, since it allows for the existence of multiple truths (but not necessarily).
Our proposed method also works for entities with a single-truth because it can automatically decide the number of truths.
Although one can easily extend a single-truth approach to handle multi-truth applications by setting a threshold (\ie, consider all values with predicted probabilities over $\lambda$ as true), it is hard to find a threshold that works for all entities.
We discuss a slightly more sophisticated extension in Section~\ref{multi:sec:hybrid}.

\section{Definitions and Notations}
\label{multi:sec:notation}

\noindent{\bf Data item, value, source, observation.}
We call an (entity, attribute) pair a {\em data item}, denoted by $d$.
Then we have a set $\cal S$ of sources provide values on $d$.
Let $v$ be a value provided by a source $s \in {\cal S}$ for the data item $d$, the pair $(d, v)$ is then called an {\em observation} of $s$. 
For instance, there are two data items in Table~\ref{multi:tbl:triples}: $d_1 =$ (ice hockey, equipments) and $d_2 =$ (snowboarding, equipments);
there are 4 values provided  for $d_1$ and 2 values provided for $d_2$.
In total we have 6 observations made by 3 sources $\{s_1, s_2, s_3\}$.

Given a data item $d$, we use $\Psi$ to denote the set of observations made on $d$ (we dismiss $d$ in the notation for simplicity); then $\Psi(s)$ denotes the values from $s$.
For example in Table~\ref{multi:tbl:triples}, for the item (ice hockey, equipments), $\Psi(s_1) = \{ {\rm helmet, stick} \}$.
Table~\ref{multi:tbl:notation} summarizes the notations we use in this chapter. 

\begin{table} 
	\centering
	\caption{Table of notations.}
	\begin{tabular}{ l l } \hline
		Notation & Description \\ \hline
		$d$ & a data item \\ 
		$v$ & a value  \\ 
		$s$ & a source that provides values \\ 
		${\cal S}_v$ & the set of sources that provide the value $v$ \\ 
		$\Psi$ & the mapping between sources and their provided values\\ 
		$\Psi(s)$ & the set of values provided by $s$ on a data item\\ 
		$n$ & the number of wrong values in the domain\\
		$\observe$ & a sequence of values that have been selected as truths\\
		$\perp$ & ``there is no more truth'' \\	\hline
	\end{tabular} 
	\label{multi:tbl:notation}
\end{table}

\smallskip \noindent
{\bf Problem definition.}
{\em Given a data item $d$ and a set ${\cal S}$ of sources, let $\cal V$ denote the set of values provided by $\cal S$ on $d$.
Our goal is to compute a probability $p(v | \Psi)$ for each value $v \in \cal V$ being true based on the source observations. }\smallskip

In this chapter we focus on the case where
the sources are independent of each other; we can extend our model with techniques 
from~\cite{accu,precrec} to address correlations between sources. 

\section{A Hybrid Model}
\label{multi:sec:hybrid}

This section presents a truth discovery model, called {\sc Hybrid}, which allows for the existence of multiple truths.
Essentially, {\sc Hybrid} makes two decisions for a data item $d$: (i) how many truths there are, and
(ii) what they are. 

One can imagine a natural solution that processes in two steps: 
(1) decide the number of truths $k$ with a single-truth method,
treating ``the number of truths for $d$'' as a data item and
$|\Psi(s)|$ as the value provided by $s$; 
(2) apply the single-truth method on $\cal V$ and select the values with top-$k$ probabilities as the truths. 
Although this approach often outperforms both the existing single-truth and multi-truth approaches (as we shall show later in Section~\ref{multi:sec:exp}), 
it has two problems. First, it does not update the value probabilities
according to its belief of the number of truths (all probabilities still sum up to 1). Second, separating the decisions into two steps may hurt precision when many sources provide more values than the truths: once the first step decides the number of truths $k$,
the second step will fill in $k$ values, possibly with values lacking strong support from the sources.

Different from the above baseline approach, {\sc Hybrid} combines the two steps and finds the truths one by one.
Conditioning on a sequence $\observe$ of true values that have been selected previously, 
it decides (1) the probability of a value $v$ being the next truth, denoted by $p(v|\observe, \Psi)$, 
and (2) the probability that there is no more truth, denoted by $p({\perp}|\observe, \Psi)$.
These are disjoint decisions so their probabilities sum up to 1. 
Thus, when {\em selecting the next truth}, {\sc Hybrid} basically applies a single-truth method. 

However, when deciding {\em whether there is any more truth} (\ie, $p({\perp}|\observe, \Psi)$),
{\sc Hybrid} incorporates the {\em unknown semantics} used in multi-truth approaches: 
if a source provides 2 values for an item $d$, it claims that it knows 2 values of $d$, 
instead of claiming that $d$ has {\em only} 2 values. 

In this way, {\sc Hybrid} combines the flexibility of the multi-truth methods of allowing multiple truths for a data item, 
and the inherent strength of the single-truth methods of considering conflicts between values
as important evidence for ruling out wrong values.
Therefore, it obtains both high precision and high recall.

Moreover, {\sc Hybrid}  leverages the typical number of truths for each type of data items;
for example, a person typically has 2 parents and 1-5 children. {\sc Hybrid} allows incorporating 
such knowledge as the {\em a priori} probability of $p({\perp}|\observe, \Psi)$, which 
further improves performance. Bear in mind that  {\em a priori} probabilities have much less
effect than observations on computing {\em a posteriori} probabilities, so {\sc Hybrid} applies 
a {\em soft} constraint rather than a {\em hard} one.

We next describe the {\sc Hybrid} model in more details, and answer the following question: {\em 
	since there should not exist any ordering between the truths, how would {\sc Hybrid} avoid the consequence of finding the truths one by one}?

\subsection{Overall Probability of a Value}
\label{multi:sec:model}

Consider computing $p(v|\Psi)$ for a value $v \in {\cal V}$. As we select truths one by one, there can be various sequences
of truths (of any length below $|{\cal V}|$) that are selected before $v$. We call each sequence $\observe$
a {\em possible world} and denote by $\Omega$ all possible worlds. Then
the probability of $v$ is the weighted sum of its probability in each possible world:
\begin{equation}
p(v | \Psi) = \sum_{\observe \in \Omega} p(v | \observe, \Psi) \cdot p(\observe | \Psi).
\label{multi:eq:valuepr}
\end{equation}
where $p(\observe | \Psi)$ is the probability of entering the possible world $\observe$.

Let $\observe = v_1v_2\dots v_{|\observe|}$, $v \notin \observe$,
denote a possible world with the sequence $v_1, v_2,\dots, v_{|\observe|}$ of values selected as truths.
Let $\observe_j$ denote a prefix of $\observe$ with length $j$ and $\observe_0=\varnothing$. 
Applying the chain rule leads us to:
\begin{equation}
p(\observe | \Psi) = \prod_{j=1}^{|\observe|} p(v_j | \observe_{j-1}, \Psi).
\label{multi:eq:chainRule}
\end{equation}

Now the only piece missing from Eq.s~\eqref{multi:eq:valuepr}-\eqref{multi:eq:chainRule} is the conditional probability
$p(v | \observe, \Psi)$, which we describe in the next subsection.

Back to the question we asked previously, even though {\sc Hybrid} finds truths one by one, 
it is order-independent as it considers all possible ways to select a value and computes an overall probability.
Apparently enumerating all possible worlds is prohibitively expensive;
we describe a polynomial-time approximation in Section~\ref{multi:sec:greedy}.

\subsection{Conditional Probability of a Value}
\label{multi:sec:cond_pr}
Now consider computing $p(v|\observe, \Psi)$.
Under a possible world $\observe$, we either choose one of the remaining values as the next truth or decide that there is no more truth, thus 
$\sum_{v' \in {\cal V}\setminus\observe}p(v'|\observe, \Psi) + p({\perp}|\observe, \Psi) = 1$;
this is similar to what we have in single-truth approaches.
Then according to the Bayes Rule, we have
\begin{equation} 
\label{multi:eq:bayes}
p(v | \observe, \Psi) {=} {p(\Psi | \observe, v) p(v|\observe) \over 
	\sum\limits_{v' \in {\cal V}\setminus\observe}p(\Psi|\observe, v')p(v'|\observe) + p(\Psi|\observe, \perp)p({\perp}|\observe)}.
\end{equation}
Here the inverse probability $p(\Psi | \observe, v)$ is the probability of observing $\Psi$ if $v$ is the next truth.
The {\em a priori} probability $p(v|\observe)$ is the probability of $v$ being the next truth regardless of the observations $\Psi$.
The same applies to $p(\Psi | \observe, \perp)$ and $p({\perp}|\observe)$.

Before we can compute these two sets of probabilities, we first define the metrics that are used to measure the quality of a source.

\subsubsection{Source-quality metrics} 
Imagine that there are $m$ latent slots for truths of a data item, and a source $s$ is asked to fill the slots.
The number of slots is unknown to $s$, so it iteratively performs two tasks: predict whether there exists another slot (\ie, another truth),
and if so, fill the slot with a value. 
We thus capture the quality of a source with two sets of metrics:
that for deciding whether there exists a truth, and that for deciding the true values.

The first set of metrics enables the {\em unknown semantics} for multi-truth, and it includes two measures:
\begin{itemize}
	\item {\em Precision} $P(s)$, the probability that when $s$ provides a value, there indeed exists a truth;
	\item {\em Recall} $R(s)$, the probability that when there exists a truth, $s$ provides a value.
\end{itemize}
Note that our $P(s)$ and $R(s)$ are different from the same notions in \cite{precrec}:
we only measure how well $s$ predicts whether or not there exists a truth, but not how well $s$ predicts what the truth is; 
in other words, we do not require the value provided by $s$ to be the same as the truth. 
To facilitate later computations, we next derive the {\em false positive rate} of $s$, denoted by $Q(s)$, from $P(s)$ and $R(s)$ by applying the Bayes Rule (see \cite{precrec} for details):
\begin{equation}
\label{multi:eq:prq}
Q(s)={\alpha \over 1-\alpha} \cdot {1-P(s) \over P(s)} \cdot R(s),
\end{equation}
where $\alpha$ is the {\em a priori} probability\footnote{\small Previous work~\cite{precrec}
	has shown that {\em a priori} probabilities play a minor role on final results comparing with the source observations.}
that a provided value corresponds to a truth slot.
Intuitively, $Q(s)$ is the probability that $s$ still provides a value when there is no truth slot.

The second set of metrics follows single-truth models to address the conflicts between values. 
It contains one measure: {\em accuracy} $A(s)$, the probability that a value provided by $s$ for a ``real'' truth slot is true (\ie, $s$ provides a true value after it has correctly predicted the existence of a truth slot).
Note that values provided for non-existing truth slots, which are absolutely false, are not counted here,
as they have been captured by $P(s)$. 

We describe how we compute these metrics in the next subsection, and demonstrate the basic idea of them
in the example below.

\begin{example}
	\label{multi:ex:src_quality}
	Consider the source $s_2$ and the data item $d_1 =$ (ice~hockey, equipments) in Table~\ref{multi:tbl:triples}.
	Suppose ice hockey requires 3 equipments . 
	We observe that $s_2$ provides 2 values on $d_1$, meaning that it predicts that
	there are 2 slots for truths; among the provided values one is true. 
	Therefore for this particular data item, $s_2$ has precision $2/2=1$, recall $2/3=0.67$, and accuracy $1/2 = 0.5$. 
	
	Now consider data item $d_2 =$ (snowboarding, equipments), which has 1 truth. 
	As $s_2$ provides 1 correct value, its precision, recall, and accuracy for this item are all 1. 
	
	If $s_2$ provides only these 2 data items, on average, we have
	$P(s_2)=\frac{1+1}{2} = 1, R(s_2)=\frac{0.67+1}{2} = 0.83$, and $A(s_2)=\frac{0.5+1}{2} = 0.75$. 
\end{example}

\subsubsection{Inverse probabilities} 
We are now ready to derive the inverse probabilities $p(\Psi | \observe, v)$ and $p(\Psi | \observe, \perp)$ in Eq.~\eqref{multi:eq:bayes}. 
Assuming that the set of sources are independent, we have 
\begin{equation}
p(\Psi | \observe, v) = \prod_{s \in{{\cal S}}}p(\Psi(s) | \observe, v),
\label{multi:eq:sro_inde}
\end{equation}
and similar for $p(\Psi | \observe, \perp)$.
In the following computations, when conditioning on $(\observe, v)$, we think that $\observe \cup \{v\}$ are the {\em only} set of truths;
similarly, when conditioning on $(\observe, \perp)$, we think $\observe$ is the only set of truths.
This is known as the {\em closed-world} assumption, and according to~\cite{online}, it should give the same results as the {\em open-world} assumption where the truths form a {\em superset} of $\observe \cup \{v\}$.

Let $\bar T$ be the truths of the item $d$, that is, $\bar T={\cal O}$ (when computing $p(\Psi | \observe, \perp)$) or $\bar T = {\cal O} \cup \{v\}$ (when computing $p(\Psi | \observe, v)$). 
Accordingly, we can partition $\Psi(s)$, values provided by $s$ on $d$,
into four categories: {\em consistent values}, {\em inconsistent values}, {\em extra values} and {\em missing values}.
We denote respectively the size of each category as $N_c, N_w, N_e, N_m$,
and the probability that a value falls into each category as $P_c, P_w, P_e, P_m$. 
Then $p(\Psi(s) | \observe, v)$ is given by:
\begin{equation}
\label{multi:eq:spellout}
p(\Psi(s) | \observe, v) = P_c^{N_c} \cdot P_w^{N_w} \cdot P_e^{N_e} \cdot P_m^{N_m}.
\end{equation}

When deriving $p(\Psi(s) | \observe, {\perp})$, 
the only difference is that we will award a source $s$ if it does not provide any extra value ; otherwise, we re-use Eq.~\eqref{multi:eq:spellout}.
The probability of not providing extra values is  $P_{\neg e} = 1 - Q(s)$, and recall that $\observe$ is the (estimated) set of truths for the data item. Thus we have:
\begin{equation} 
\label{multi:eq:spell_perp}
p(\Psi(s) | \observe, {\perp}) = \begin{cases}
P_c^{N_c} \cdot P_w^{N_w} \cdot P_e^{N_e} \cdot P_m^{N_m} & |\Psi(s)| > |\observe|; \\
P_c^{N_c} \cdot P_w^{N_w} \cdot P_e^{N_e} \cdot P_m^{N_m} \cdot P_{\neg e} & |\Psi(s)| \leq |\observe|.
\end{cases}
\end{equation}

We next define each category and describe how we compute their sizes and probabilities.
Following~\cite{accu}, we assume that there are $n$ false values in the domain of $d$ and they are uniformly distributed (note that the false values may not appear in $\cal V$).
\begin{itemize}
	\item {\em Consistent value:} A consistent value is a value in $\bar T \cap \Psi(s)$; thus, $N_c=|\bar T \cap \Psi(s)|$.
	To provide a consistent value, $s$ needs to correctly predict that there exists a slot for truth,
	and fills the slot with a true value, so $P_c=R(s)\cdot A(s)$.
	
	\item {\em Inconsistent value:} An inconsistent value is a value that is provided for a truth slot, but differs from any true value.
	At most we have $|\bar T|$ values provided for truth slots;
	except the consistent values, others are inconsistent.
	Thus $N_w=\min(|\bar T|, |\Psi(s)|) - N_c$.
	When $s$ provides an inconsistent value, it correctly predicts the existence of a truth slot, 
	but fills in a {\em particular} false value, so $P_w=R(s){\cdot} {1-A(s) \over n}$. 
	
	\item {\em Extra value:} If $s$ provides more than $|\bar T|$ values, the rest of the values are extra values;
	thus, $N_e = \max(|\Psi(s)|-|\bar T|, 0)$. When $s$ provides an extra value, it incorrectly predicts
	a non-existing slot, and fills in a {\em particular} (false) value, so 
	$P_e={Q(s) \over n}$.
	
	\item {\em Missing value:} Alternatively when $\Psi(s)$ contains fewer values than $\bar T$, $s$ misses some truth slots (\ie, $s$ thinks they do not exist). We have $N_m = \max(|\bar T|-|\Psi(s)|, 0)$ and $P_m = 1-R(s)$. 
	
\end{itemize}

\begin{example}
	\label{exa:inverse_pr}
	Consider the data item $d_1$ and we now compute 
	$p(\Psi(s_2)|o_1, o_2)$, the probability of observing the values in $\Psi(s_2)$ if $o_2$ is the next truth after $o_1$ has been selected.
	We have $\observe = o_1$, $\Psi(s_2) = \{o_2, o_3\}$ and $\bar{T} = \{o_1, o_2\}$.
	So $\Psi(s_2)$ contains one consistent value, and one inconsistent value; there is no extra value or missing value. In other words, we have $N_c=N_w=1$ and $N_e=N_m=0$.
	
	Supposing $n=10$ and $A(s_2)=0.6, R(s_2)=0.9, Q(s_2)=0.1$, we have 
	$P_c=0.9\cdot 0.6=0.54$, $P_w = 0.9 \cdot \frac{1-0.6}{10} = 0.036$, $P_e = \frac{0.1}{10} = 0.01$, $P_m=0.1$.
	
	Then according to Eq.~\eqref{multi:eq:spellout} we compute:\\
	$p(\Psi(s_2)|o_1, o_2) = 0.54^1 \cdot 0.036^1 \cdot 0.01^0 \cdot 0.1^0 = 0.019$.
	
	We repeat the above process for the other sources and obtain: \\
	$p(\Psi(s_1)|o_1, o_2) = 0.292$, $p(\Psi(s_3)|o_1, o_2) = 0.019$.
	
	With the source-independence assumption, we have: \\
	$p(\Psi|o_1, o_2) = 0.019\cdot 0.292 \cdot 0.019 = 1.05 \times 10^{-4}$.
	
\end{example}

\subsubsection{A priori probabilities}
We then compute the probabilities $p({\perp}|{\cal O})$ and $p(v|{\cal O})$ in Eq.~\eqref{multi:eq:bayes}.
Intuitively, the chance of $\perp$ increases when more truths are found.
Let $\beta_i$ be the {\em a priori} probability of $\perp$ when we are looking for the $i$-th truth (\ie, $|\observe| = i-1$).
So there are $ |{\cal V}|-i+1$ unselected values in $\cal V$; 
assuming they have the same {\em a priori} probability, the {\em a priori} probability $p(v|{\cal O})$ of each value $v$ would be:
\begin{equation}
p(v|{\cal O}) = \frac{1-\beta_i}{|{\cal V}|-i+1}. 
\label{multi:eq:prior}
\end{equation}

We can derive $\beta_i$ from the distribution
of the number of truths for a data item. For example, among people who have children, if $30\%$ of them have 1 child, $40\%$ have 2 children, and so on,
then $\beta_2=0.3$
(with probability $30\%$ there is not a second truth), and $\beta_3=0.7$ (with probability $30\%$+$40\%$
there is not a third truth). 

\subsubsection{Summary}
By putting the derived inverse probabilities and {\em a priori} probabilities into Eq.~\eqref{multi:eq:bayes}, we are able to obtain $p(v|\observe,\Psi)$, and this completes the computation of Eq.~\eqref{multi:eq:valuepr}.
As the following proposition shows, {\sc Hybrid} computes higher probabilities for values provided
by more accurate sources; it finds more truths when high-precision sources provide more values;
and it finds fewer truths when high-recall sources provide fewer values. These all conform to our intuition.

\begin{proposition}
	\label{multi:prop:feature}
	Consider a value $v$ and a source $s \in{{\cal S}}$ where $v \in \Psi(s)$; we fix all sources in ${{\cal S}}$ except $s$.
	\begin{itemize}
		\item If $A(s) > \frac{1}{u+1}$, $p(v|\Psi)$ increases when $A(s)$ increases.
		
		\item If $Q(s) < \frac{R(s) - R(s)A(s)}{1-R(s)A(s)}$, $p({\perp}|\Psi)$ decreases when $s$ provides more values.
		
		\item If $R(s) > \frac{Q(s)}{1-A(s)+A(s)Q(s)}$, $p({\perp}|\Psi)$ increases when $s$ provides fewer values. \rbox
	\end{itemize}
\end{proposition}

\begin{example}
	Continuing with Example~\ref{exa:inverse_pr}, we proceed to compute $p(o_2|o_1, \Psi)$ using Eq.~\eqref{multi:eq:bayes}.
	This requires the inverse probability $p(\Psi|o_1,v)$ and the {\em a priori} probability $p(v|o_1)$ for every remaining value in ${\cal V}\setminus \observe = \{o_2, o_3, o_4\}$ as well as $\perp$.
	
	We have obtained the inverse probability $p(\Psi|o_1,o_2)$ in Example~\ref{exa:inverse_pr}; we now repeat the process on $o_3$, $o_4$ and $\perp$ to compute:\\
	\indent $p(\Psi|o_1, o_3) = p(\Psi|o_1, o_4) = 6.8\times 10^{-6}$; \\
	\indent $p(\Psi|o_1, \perp) = 1.05\times 10^{-8}$.\\
	Then assuming $\beta_1 = p({\perp}|o_1) =0.3$, from Eq.~\eqref{multi:eq:prior} we have \\
	\indent $p(o_2|o_1) = p(o_3|o_1) = p(o_4|o_1) = \frac{1-\beta_1}{|{\cal V}|-1+1} = 0.175$.\\
	We can thus obtain $p(o_2|o_1, \Psi)$ via Eq.~\eqref{multi:eq:bayes}:\\
	$p(o_2|o_1, \Psi) =\frac{p(\Psi|o_1,o_2) p(o_2|o_1)}{\sum_{v \in \{o_2, o_3, o_4\}} p(\Psi|o_1,v) p(v|o_1) + p(\Psi|o_1,\perp) p({\perp}|o_1)} \\= 0.88$.
	
	Table~\ref{multi:tbl:prob_example} shows the probabilities obtained by enumerating all possible worlds $\observe$ for each value $v$.
	We can see that {\sc Hybrid} gives very high probabilities (0.92) for the two true values ({\em helmet} and {\em stick}) and meanwhile very low probabilities for the false ones. 
	
\end{example}

\subsection{Evaluating Source Quality}

The previous subsection explains how to compute value probabilities based on the quality of sources.
We do not always have such prior knowledge on source qualities, and in this case we start by assuming each source has the same quality, and then iteratively compute value probabilities and source qualities until convergence.
This subsection describes how to update source quality based on the estimated truths of a set of data items.

For each source $s$, we compute $P(s), R(s)$ and $A(s)$ as defined, except that we adopt the probabilistic
decisions made on the truthfulness of values.
We emphasize again that the computation of precision and recall do not consider the {\em truthfulness}
of the values, but only the {\em cardinality} of the provided values (\ie, how many truth slots 
$s$ thinks there are).
\begin{itemize}
	\item The precision of $s$ is the average of its precision on each data item $s$ provides values for.
	Let $\Psi_d(s)$ be the set of values provided by $s$ on $d$ and ${\cal V}_d$ be the domain of $d$.
	Then $\sum_{v {\in} {\cal V}_d} p(v)$ is the (probabilistic) number of truths for $d$, and we have
	\begin{equation} 
	\label{multi:eq:precision}
	P(s) = \avg_{d} ~\min ( \frac{\sum_{v {\in} {\cal V}_d} p(v|\Psi)}{|\Psi_d(s)|}, 1 ).
	\end{equation}
	
	\item Similarly, the recall of $s$ is the average of its recall on each data item.
	\begin{equation} 
	\label{multi:eq:recall}
	R(s) = \avg_{d} ~\min ( \frac{|\Psi_d(s)|}{\sum_{v {\in} {\cal V}_d} p(v|\Psi)}, 1 ).
	\end{equation}

	\item The accuracy of $s$ can be estimated as the average probability of its values,
	divided by its precision, so it accounts for values provided for ``real'' truth slots only.
	\begin{equation} 
	A(s) = {\avg_{d, v \in \Psi_d(s)} p(v|\Psi) \over P(s)}.
	\end{equation}  
\end{itemize}

\section{Approximation for HYBRID}
\label{multi:sec:greedy}

Computing value probabilities by enumerating all possible worlds 
takes exponential time. We conjecture that the value probability computation in
{\sc Hybrid} is \#P-complete; the proof of the conjecture remains an open problem.
This section describes an approximation for probabilities under {\sc Hybrid}.
We start with simplifying the computation of $p(v|{\cal O}, \Psi)$ in Eq.~\eqref{multi:eq:bayes}, and then present
our approximation algorithm.

\subsection{Simplification of $p(v|{\cal O}, \Psi)$} 
We can simplify the computations in Section~\ref{multi:sec:cond_pr} to a much simpler form.
We start from Eq.~\eqref{multi:eq:spellout} and Eq.~\eqref{multi:eq:spell_perp}.

Given a particular source $s$, suppose $N_c = c$, $N_w=w$, $N_e=e$ and $N_m = m$ when computing $p(\Psi(s) | \observe, {\perp})$ with Eq.~\eqref{multi:eq:spell_perp}.
Then we have four cases when deciding the $N_c$, $N_w$, $N_e$ and $N_m$ for $p(\Psi(s) | \observe, v)$ in Eq.~\eqref{multi:eq:spellout}, depending on whether $|\Psi(s)| > |\observe|$ and whether $v \in \Psi(s)$; we illustrate them in Table~\ref{tbl:num_values}.

\begin{table}
	\setlength{\tabcolsep}{4.5pt}
	\centering
	\caption{Numbers used in Eq.~\eqref{multi:eq:spellout} when $N_c = c, N_w=w, N_e=e$ and $N_m = m$ in Eq.~\eqref{multi:eq:spell_perp}.}
	\begin{tabular}{@{}c|l|l|l|l@{}} \hline
		The case that $s$ belongs to & $N_c$ & $N_w$ & $N_e$ & $N_m$ \\ \hline
		case 1: $|\Psi(s)| > |\observe|$ and $v \in \Psi(s)$ & $c+1$ &  $w$& $e-1$ &  $m$\\ 
		case 2: $|\Psi(s)| > |\observe|$ and $v \notin \Psi(s)$ & $c$ &  $w+1$& $e-1$ &  $m$\\ 
		case 3: $|\Psi(s)| \leq |\observe|$ and $v \in \Psi(s)$ & $c+1$ &  $w-1$& $e$ &  $m+1$\\ 
		case 4: $|\Psi(s)| \leq |\observe|$ and $v \notin \Psi(s)$ & $c$ &  $w$& $e$ &  $m+1$\\ \hline
	\end{tabular} 
	\label{tbl:num_values}
\end{table}

Let ${\cal S}_1$, ${\cal S}_2$, ${\cal S}_3$, ${\cal S}_4$ denote the set of sources in ${\cal S}$ that fall into each of the above cases respectively.
We can then write $p(\Psi| \observe, v)$ and $p(\Psi | \observe, {\perp})$ into the following formats:
\begin{align}
p(\Psi| \observe, v) &= \prod_{\bar {\cal S}_1 \cup \bar {\cal S}_3} (P_c)^{c+1} \cdot
\prod_{\bar {\cal S}_2 \cup \bar {\cal S}_4} (P_c)^{c} \cdot \label{multi:eq:v_full}\\
&~~~~~~~\prod_{\bar {\cal S}_1} (P_w)^w \cdot
\prod_{\bar {\cal S}_2} (P_w)^{w+1} \cdot
\prod_{\bar {\cal S}_3} (P_w)^{w-1} \prod_{\bar {\cal S}_4} (P_w)^{w} \cdot \nonumber\\
&~~~~\prod_{\bar {\cal S}_1 \cup \bar {\cal S}_2} (P_e)^{e-1}(P_m)^{m} \cdot
\prod_{\bar {\cal S}_3 \cup \bar {\cal S}_4} (P_e)^{e}(P_m)^{m+1};\nonumber\\
p(\Psi | \observe, {\perp}) &= \prod_{\bar {\cal S}} (P_c)^{c} (P_w)^{w} (P_e)^{e} (P_m)^{m} \cdot
\prod_{\bar {\cal S}_3 \cup \bar {\cal S}_4} P_{\neg e}.
\label{multi:eq:empty_full}
\end{align}

Next let 
\begin{align*}
C &= \prod\limits_{\bar {\cal S}} (P_c)^{c} \cdot
\prod\limits_{\bar {\cal S}_1 \cup \bar {\cal S}_2} (P_w)^{w+1}(P_e)^{e-1}(P_m)^{m} \cdot\\
&~~\prod\limits_{\bar {\cal S}_3 \cup \bar {\cal S}_4} (P_w)^{w}(P_e)^{e}(P_m)^{m+1};
\end{align*}
we can simplify Eq.~\eqref{multi:eq:v_full} and Eq.~\eqref{multi:eq:empty_full} as follows:
\begin{align}
& p(\Psi| \observe, v) = C \cdot \prod_{\bar {\cal S}_1 \cup \bar {\cal S}_3} \frac{P_c}{P_w};\\
& p(\Psi | \observe, {\perp}) = C \cdot 
\prod_{\bar {\cal S}_1 \cup \bar {\cal S}_2} \frac{P_e}{P_w} \cdot 
\prod_{\bar {\cal S}_3 \cup \bar {\cal S}_4} \frac{P_{\neg e}}{P_m}.
\end{align}

Next, we define the {\em vote count} of a value $v$ based on the accuracy of its providers:
\begin{align}
L(v) &= \prod_{\bar {\cal S}_1 \cup \bar {\cal S}_3} \frac{P_c}{P_w} 
= \prod_{\bar {\cal S}_1 \cup \bar {\cal S}_3} \frac{R(s)A(s)}{R(s)\frac{1-A(s)}{n}} \nonumber \\
&= \prod_{S \in{\bar{\cal S}}_v} \frac{nA(s)} {1-A(s)};
\end{align}

The vote count of $\perp$ at the $i$-th step (\ie, $i-1$ truths have been selected) combines the {\em a priori} probability and the votes from all sources: 
\begin{align}
L_{|\observe|+1}(\perp) &= \frac{p({\perp}|\observe)}{p(v|\observe)} \cdot
\prod_{\bar {\cal S}_1 \cup \bar {\cal S}_2} \frac{P_e}{P_w} \cdot
\prod_{\bar {\cal S}_3 \cup \bar {\cal S}_4} \frac{P_{\neg e}}{P_m} \nonumber\\
&= {\beta_i(|{\cal V}|{-}i{+}1) \over 1-\beta_i} \cdot
\prod_{|\Psi(s)| > |\observe|} {Q(s) \over R(s)(1-A(s))} \cdot\\
&~~~~ \prod_{|\Psi(s)| \leq |\observe|} {1-Q(s) \over 1-R(s)} \nonumber.
\end{align}

We  can then transform Eq.~\eqref{multi:eq:bayes} into the following format:
\begin{equation}
\label{multi:eq:value_cond}
p(v|\observe, \Psi)={L(v) \over \sum_{v' \in {\cal V}\setminus \observe} L(v')+L_{|\observe|+1}(\perp)}.
\end{equation}

\subsection{Approximation} 
Our approximation leverages three observations. First, equivalent to computing $p(v|\Psi)$
conditioning on all possible worlds,
we compute $p(v|\Psi)=\sum_{i}p_i(v|\Psi)$, where $p_i(v|\Psi)$ denotes the probability of $v$
being the $i$-th truth, computed by considering possible worlds $\observe$ with $i-1$ values.
Second, although there are multiple possible worlds with size $i-1$, the nature of Bayesian analysis 
determines that one of them would have a much higher probability than the others, so we can use it for approximation. 
Third, once the probability of $\perp$ is above that of the $i$-th truth, it quickly increases to $1$
in the following steps. Therefore if we terminate at the $i$-th step, we would not lose much.
Recall that the confidence of a value $v$ does not change with $i$, but only that
of $\perp$ changes, thus we can easily decide the number of steps we need before termination.

Algorithm~\ref{multi:algo:greedy} gives the details of the approximation.
\begin{itemize}
	\item Without losing generality, let $L(v_1), L(v_2), \dots$ be a sorted list in {\em decreasing} order, that is, $L(v_{i-1}) \geq L(v_i)$ for $\forall i > 1$ (Lines~\ref{begin:order}-\ref{end:order}). 
	Let $k$ be an integer where $L(v_{k-1}) \geq L_{k-1}(\perp)$ and $L(v_k) < L_k(\perp)$;  we thus terminate after $k$ steps, and the first $k-1$ values are considered as truths (Lines~\ref{begin:terminate}-\ref{end:terminate}). 
	
	\item We first initialize $p(v|\Psi)$ to $0$ for each $v \in {\cal V}$. 
	Then at each step $i$, we update $p(v|\Psi)$ by adding the probability $p_i(v|\Psi)$ for $v$ to be the $i$-th truth  (Line~\ref{end:valuePr}).
	To compute $p_i(v|\Psi)$, we consider possible worlds where $v$ is not present yet (their probabilities sum up to $1-p(v|\Psi)$).
	Assuming $v$ has the same conditional probability in all these possible worlds, denoted by $p(v|\Gamma_i,\Psi)$, we have:
	\begin{equation}
	\label{multi:eq:upper_v}
	p_i(v|\Psi) = (1-p(v|\Psi)) \cdot p(v|\Gamma_i, \Psi).
	\end{equation}
	
	\item We obtain $p(v|\Gamma_i, \Psi)$ from the possible world with the largest probability,
	which must have selected the $i{-}1$ values with the highest confidence as truths; 
	that is, $\Gamma_i = v_1v_2\dots v_{i-1}$.
	We thus compute $p(v|\Gamma_i, \Psi)$ by normalizing the subsequence of confidence starting with $L(v_i)$ (Line~\ref{pr}).
	Note that for any possible world $\observe$ with length $i-1$, we have $p(v|\observe, \Psi) \leq p(v|\Gamma_i, \Psi)$.
\end{itemize}

\begin{algorithm} 
	\SetKwInOut{Input}{input}
	\SetKwInOut{Output}{output}
	\caption{Approximation for {\sc Hybrid}}
	\label{multi:algo:greedy}
	\Input{Observations $\Psi$ containing a set $\cal V$ of values provided by a set ${\cal S}$ of sources on data item $d$; prior probability $\beta$}
	\Output{Probability $p(v|\Psi)$ for each $v \in {\cal V}$}
	\BlankLine
	\ForEach{$v \in {\cal V}$}{ \label{begin:order}
		$p(v|\Psi) \leftarrow 0$\; \label{init}
		Compute $L(v)$ using Eq.~\eqref{multi:eq:confv}\;
	}
	Let $L(v_1), L(v_2), \dots$ be a sorted list in decreasing order\; \label{end:order}
	
	\ForEach{$i \in [1, |{\cal V}|]$}{ \label{begin:outer}
		Compute $L_i(\perp)$ using Eq.~\eqref{multi:eq:confempty}\;
		\ForEach{$v \in {\cal V}$}{   \label{begin:valuePr}
			$p(v|\Gamma_i, \Psi) = \min (\frac{L(v)}{\sum_{j=i}^{j=|\cal V|}L(v_j)+ L_i(\perp)},1)$\; \label{pr} 
			$p(v|\Psi) \leftarrow p(v|\Psi) + (1-p(v|\Psi)) \cdot p(v|\Gamma_i, \Psi)$\;   \label{end:valuePr}
		}
		\If {$L_i(\perp) > L(v_i) $} {  \label{begin:terminate}
			break\; \label{end:terminate}
		}
	}
\end{algorithm}

\begin{example}
	Consider again computing $p(o_2|\Psi)$. The sorted list of the value confidences is $\{225, 225, 15, 15\}$, given by $o_1$, $o_2$, $o_3$ and $o_4$;
	the confidences of $\perp$ in different steps are $\{0.1, 0.24, 18033, \\18033\}$.
	We thus terminate after the third step (when $i=3$). 
	
	When $i{=}1$, we compute $p_1(o_2|\Psi)=0.47$.
	
	When $i{=}2$, we compute $p(o_2|\Gamma_2,\Psi)= p(o_2|o_1,\Psi)=0.88$, thus $p_2(o_2|\Psi)=(1-0.47)\times 0.88 = 0.47$.
	
	When $i{=}3$, we compute $p(o_2|\Gamma_3,\Psi)= \frac{225}{15+15+18033} = 0.01$, thus
	$p_3(o_2|\Psi)=(1-0.47-0.47) {\times} 0.01 = 0.0006$. 
	
	The final probability for $o_2$ is $p(o_2| \Psi) = p_1(o_2|\Psi) + p_2(o_2|\Psi) + p_3(o_2|\Psi)= 0.9406$, very close to the probability $0.92$ obtained by {\sc Hybrid}.
\end{example}

The next theorem shows that Algorithm~\ref{multi:algo:greedy} approximates the value probabilities both efficiently and
effectively.
\begin{theorem}
	Let $d$ be a data item and $n$ be the number of values provided for $d$.
	\begin{itemize}
		\item Algorithm~\ref{multi:algo:greedy} estimates the probability of each provided value 
		in time $O(n^2)$.
		\item For each value $v$ on $d$, we have $|p(v) - \hat{p}(v)| < \frac{1}{6}$,
		where $\hat{p}(v)$ is the exact probability computed by {\sc Hybrid}, 
		and $p(v)$ is the probability obtained by Algorithm~\ref{multi:algo:greedy}. \rbox
	\end{itemize}
	\label{multi:the:bound}
\end{theorem}

\begin{proof}
	See Appendix~\ref{append:bound}.
\end{proof}

\section{Experimental Study}
\label{multi:sec:exp}

We now present experimental results to evaluate the proposed approach.
Section~\ref{multi:sec:exp:setting} describes experimental settings.
Then Section~\ref{multi:sec:exp:fusion} gives a comprehensive comparison of various fusion models on a widely used real dataset as well as synthetic data,
showing that {\sc Hybrid} outperforms others in general and is the most robust.

\subsection{Experimental Settings}
\label{multi:sec:exp:setting}

\noindent {\bf Methods to compare.}
We compared the following fusion algorithms.
\begin{itemize}
	\item {\sc Accu}~\cite{accu}, the single-truth model reviewed in Section~\ref{multi:sec:intro}. For each data item, it considers the value with the highest predicted probability as the truth.
	\item {\sc PrecRec}~\cite{precrec}, the multi-truth model reviewed in Section~\ref{multi:sec:intro}. It considers a value correct if its predicted probability is above 0.5.
	\item {\sc LTM}~\cite{LTM}, a multi-truth model using directed graphical model.
	It also considers a value correct if its predicated probability is above 0.5.
	\item {\sc TwoStep}, the baseline method described in Section~\ref{multi:sec:hybrid}. It first 
	decides the number of truths $k$, and then returns the $k$ values with top probabilities according to {\sc Accu}.
	\item {\sc Hybrid}, Algorithm~\ref{multi:algo:greedy} described in Section~\ref{multi:sec:greedy}. 
	It considers the values obtained before the termination step as the truths.
\end{itemize}

\noindent {\bf Implementations.} 
Whenever applicable, we set $n=10, \alpha=0.25$, and consider only ``good'' sources 
(\eg, sources on which the conditions in Proposition~\ref{multi:prop:feature} hold). 
We initialize the source quality metrics as $A=0.8, R=0.8, Q=0.2$, and then iteratively compute value probabilities and source qualities for up to 5 iterations.
We implemented all methods in Java on a MapReduce-based framework.

\smallskip
\noindent
{\bf Metrics.} We report {\em precision} and {\em recall} for each method. {\em Precision} measures among all
observations predicted as correct, the percentage that are true. {\em Recall} measures among all
true observations, the percentage that are predicted as correct. {\em F-measure} is computed as 
${2 \cdot prec \cdot rec \over prec + rec}$. (Note that they are different from {\em precision}
and {\em recall} of individual sources as defined in Section~\ref{multi:sec:cond_pr}).

\subsection{Experiment Results}
\label{multi:sec:exp:fusion}

\subsubsection{Results on Book data}

We first use the Book data from \cite{truthFinder}, which  has been widely used for knowledge-fusion experiments. 
As shown in Table~\ref{tbl:book}, it contains 6,139 book-author triples on 1,263 books from 876 retailers. 
The gold standard consists of authors for 100 randomly sampled books,
where the authors were manually identified from book covers. 
According to the gold standard, 62\% of the provided authors are correct, and 98\% of the true values are provided by some source.
There are 57\% of the books that have multiple authors.

\begin{table}
	\centering
	\fontsize{8.5}{8.5}\selectfont
	\caption{Statics of the Book data.}
	\begin{tabular}{@{}r r r r r r@{}}
		\hline \#entities & \#triples & \#sources& precision & recall & \%multi-truth  \\ 
		\hline 1,263 & 6,139 & 876 & 0.62 & 0.98 & 57\% \\ 
		\hline 
	\end{tabular} 
	\label{tbl:book}
\end{table}

In addition to the five fusion methods we listed before, we also compared with {\sc Accu\_list},
which applies {\sc Accu} but considers the full list of authors as a whole~\cite{accu, truthFinder}.

\begin{table}
	\centering
	\caption{Results on Book data. {\sc Hybrid} obtains the highest recall and F-measure.}
	\begin{tabular}{lrrr}
		\hline
		& Precision & Recall &    F1 \\ \hline
		{\sc Accu} &     \textbf{0.990} &  0.532 & 0.692 \\ 
		{\sc Accu\_list}        &     0.974 &  0.801 & 0.879 \\ \hline
		LTM               &     0.911 &  \textbf{0.973} & 0.941 \\
		{\sc PrecRec}     &     0.965 &  0.931 & 0.947 \\ \hline
		{\sc Hybrid}      &     0.941 &  \textbf{0.973} & \textbf{0.957} \\ \hline
	\end{tabular}
	\label{tbl:result_book}
\end{table}

Table~\ref{tbl:result_book} shows the results, and we can see that \model{} obtains a higher F-measure than existing single-truth and multi-truth models.
By considering both conflicts between values and the possibility of having multiple truths, 
it is able to identify more true values without sacrificing much of precision.
Not surprisingly, {\sc Accu} has the highest precision but the lowest recall as it only finds one author for a book.
{\sc LTM} has a lower precision as it lacks a global view of the values provided for the same data item.
Instead, {\sc PrecRec} has a lower recall but a higher precision:
in this dataset many sources only provide the first author of a book and this explains the low recall;
the high precision is because the sources provide few wrong values.
{\sc TwoStep} separates the decisions of how many truths there are and what they are, 
so may return authors that do not have strong support, leading to a low precision.

\subsubsection{Results on Synthetic Data}
\label{multi:sec:exp:syn}

To better understand the performance of different approaches in various situations, 
we compare them on synthetic data where we vary the number of truths and the quality of sources.

We generated 10 data sources providing values on 100 data items,
where wrong values were randomly selected from a domain of 100 values.
We varied the following parameters when generating the data.
\begin{itemize} 
	\item \textit{Number of truths} for each data item: ranges from 1 to 10, and by default follows a Gaussian distribution 
	with mean=6 and std=1.
	\item \textit{Source accuracy}: ranges from 0.2 to 1, and 0.7 by default.
	\item \textit{Source recall}: ranges from 0.2 to 1, and 0.7 by default.
	\item \textit{Extra ratio}: equals to $\frac{N_e}{N_c+N_w}$ (see Eq.~\eqref{multi:eq:spellout});
	ranges from 0.2 to 1, and 0.2 by default.
\end{itemize}
All experiments were repeated 100 times and we report the average.

\begin{figure*}
	\centering
	\includegraphics[width=0.267\textwidth]{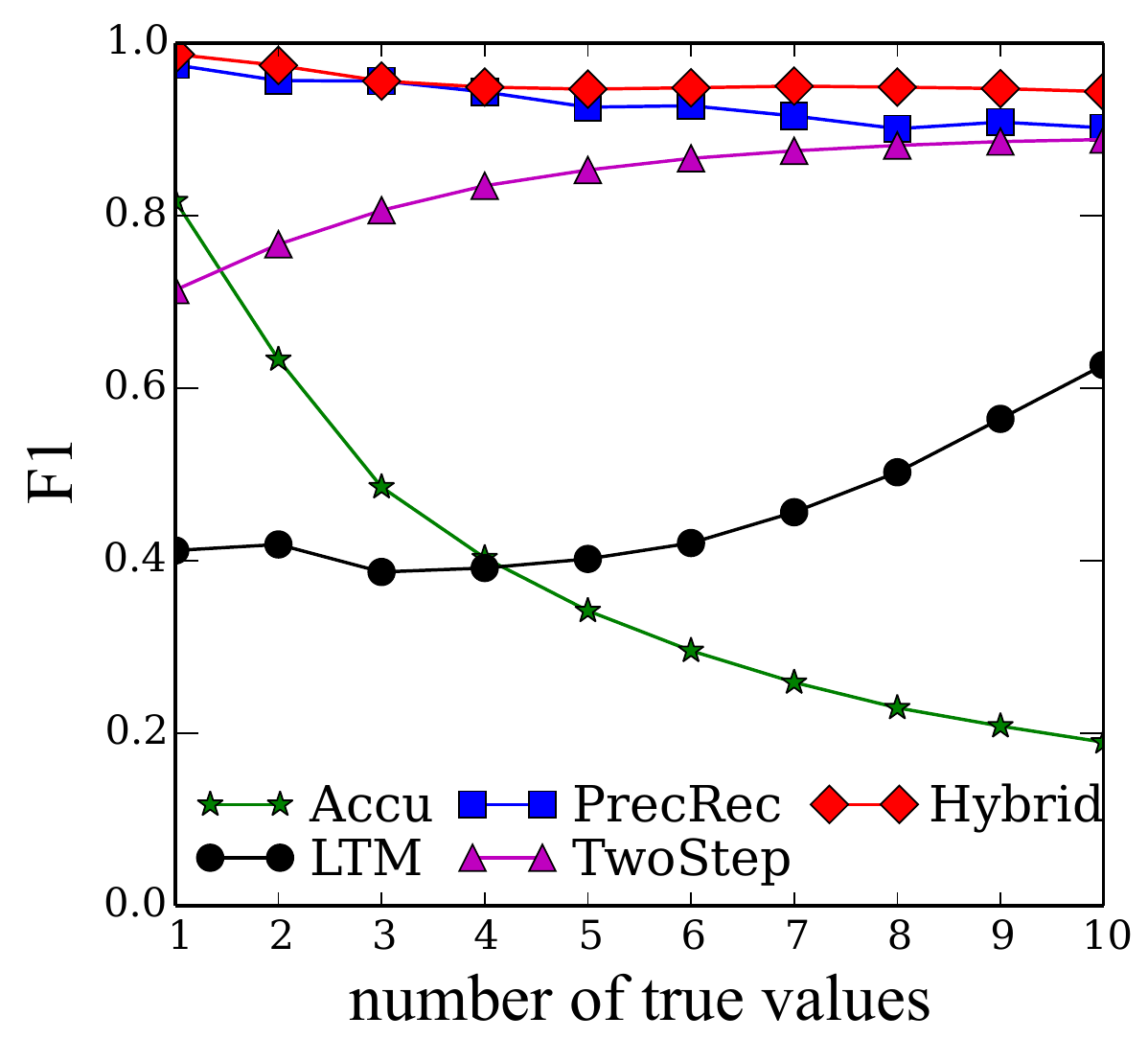} \quad\quad
	\includegraphics[width=0.267\textwidth]{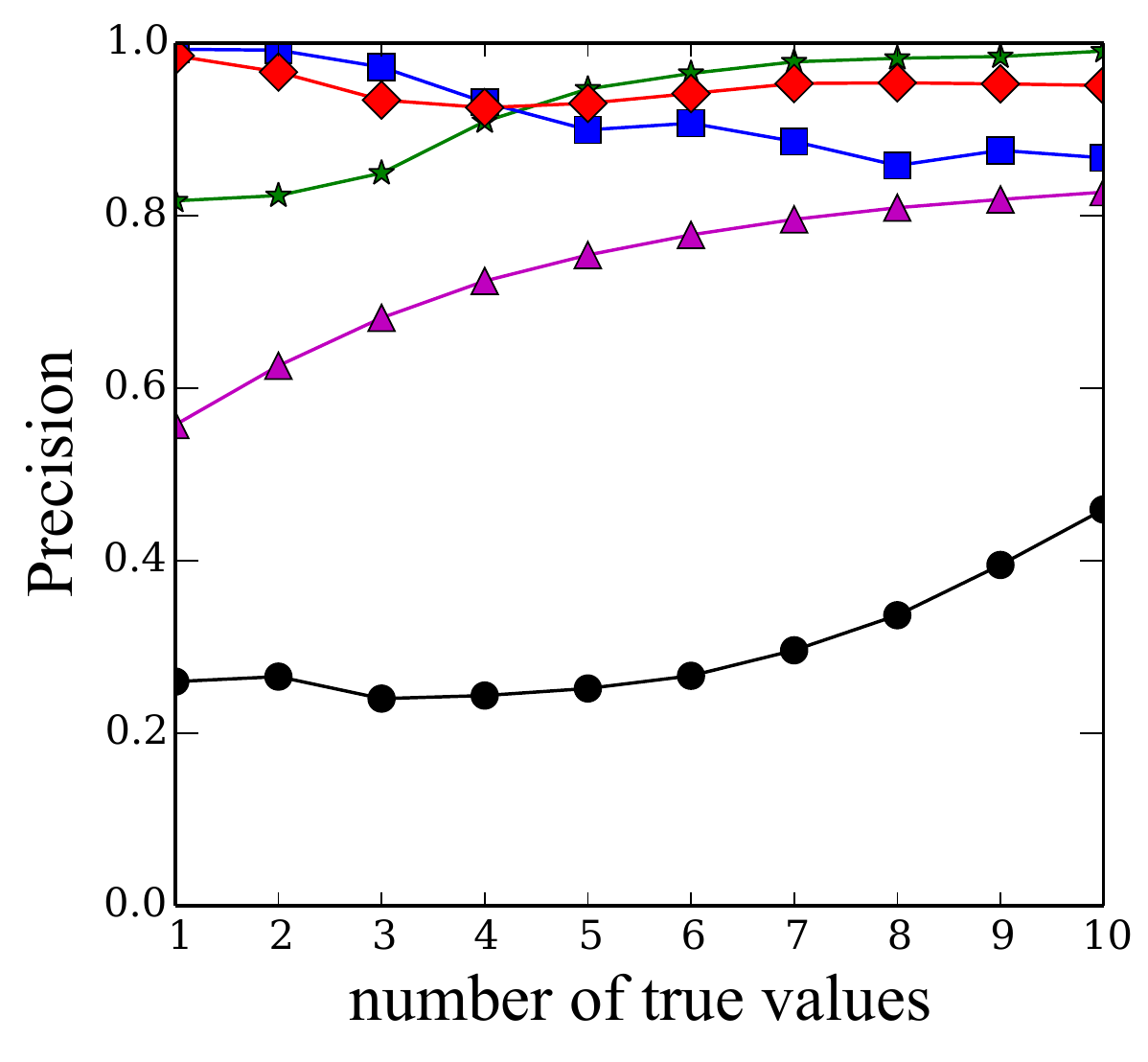} \quad\quad
	\includegraphics[width=0.267\textwidth]{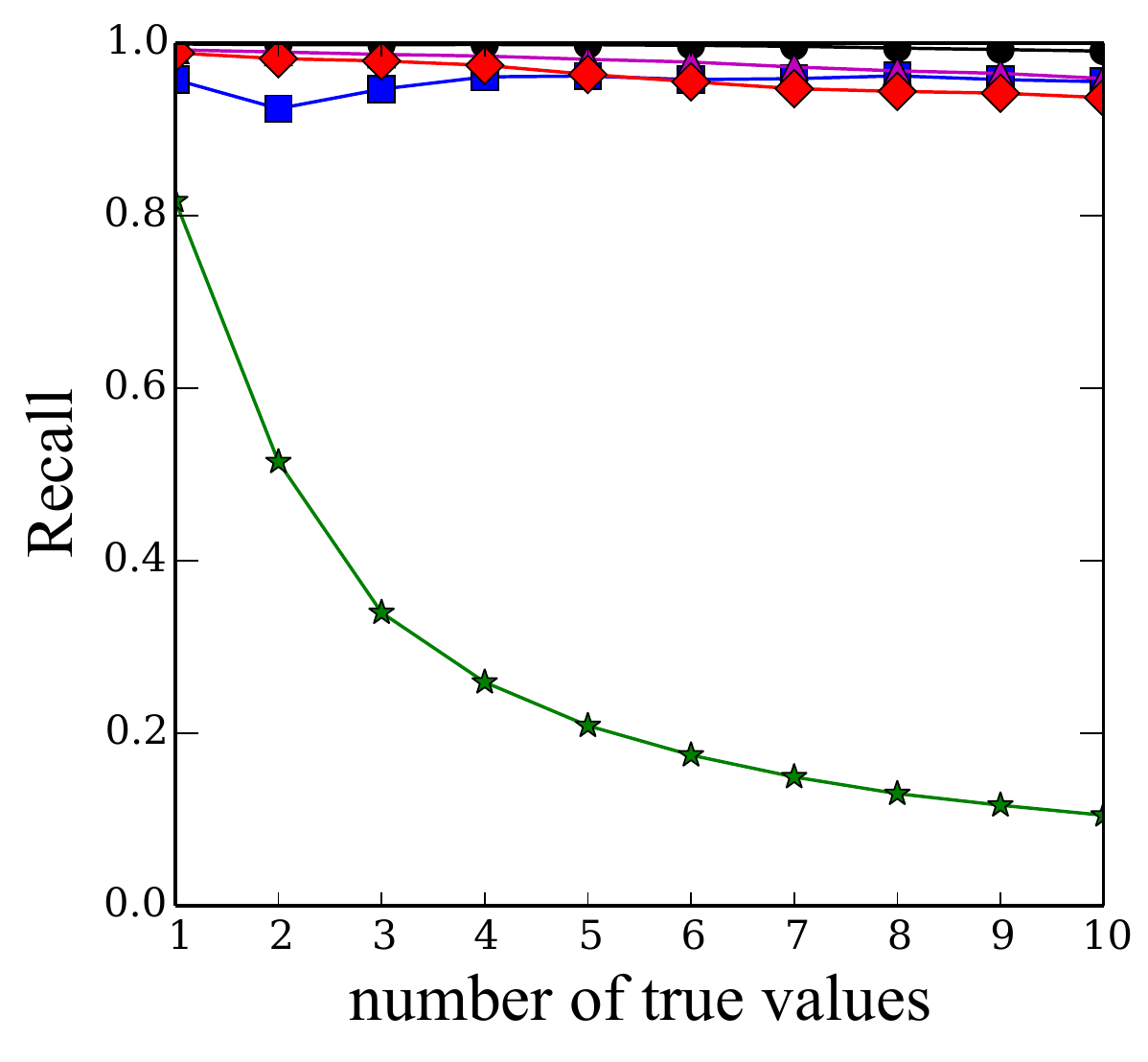}
	\caption{Varying the number of truths on synthetic data. {\sc Hybrid} improves over other models when the number of truths is large.}
	\label{multi:fig:syn_vary_truthNum}
\end{figure*}

\begin{figure*}
		\centering
		\includegraphics[width=.267\textwidth]{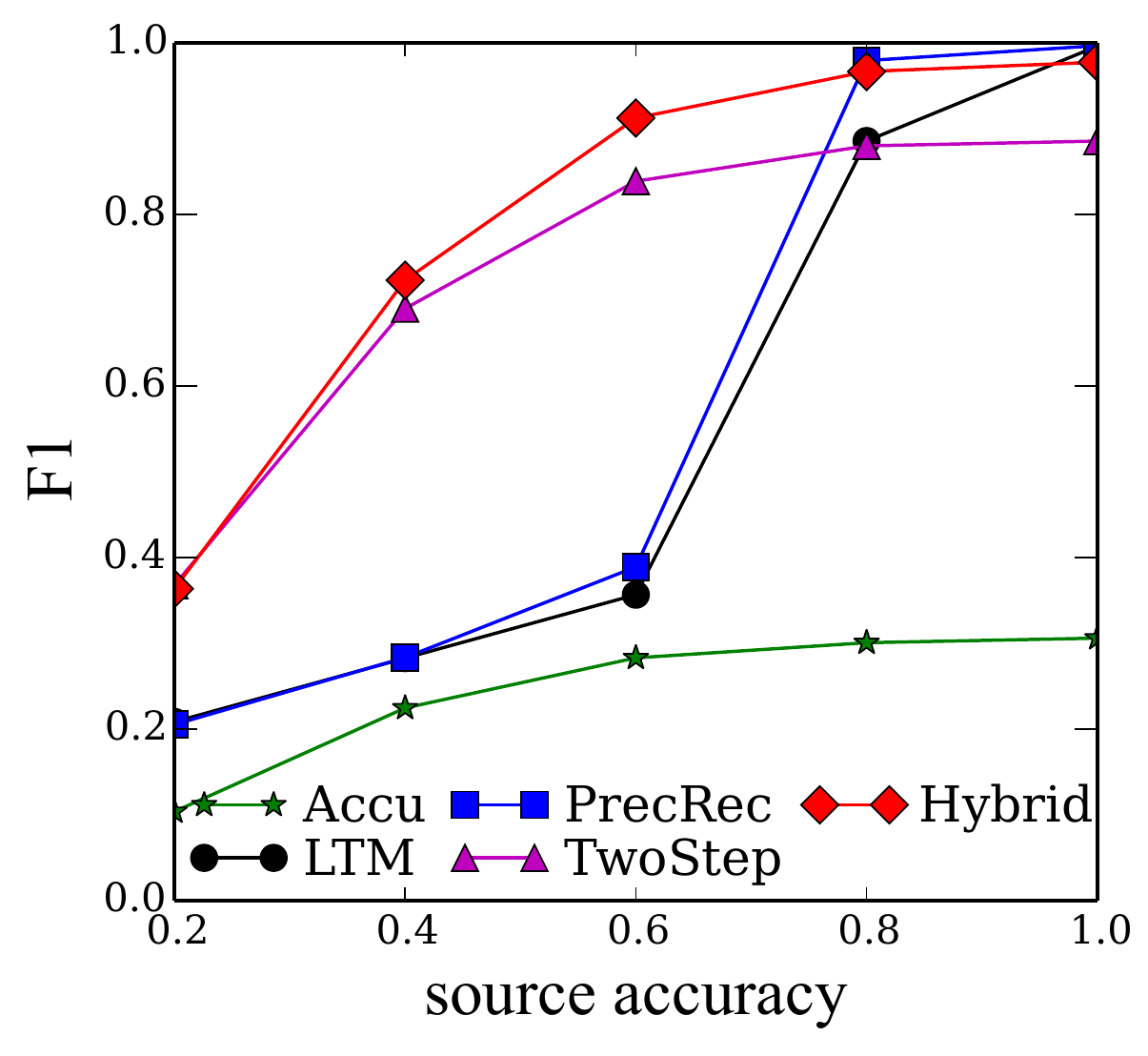} \quad\quad
		\includegraphics[width=.267\textwidth]{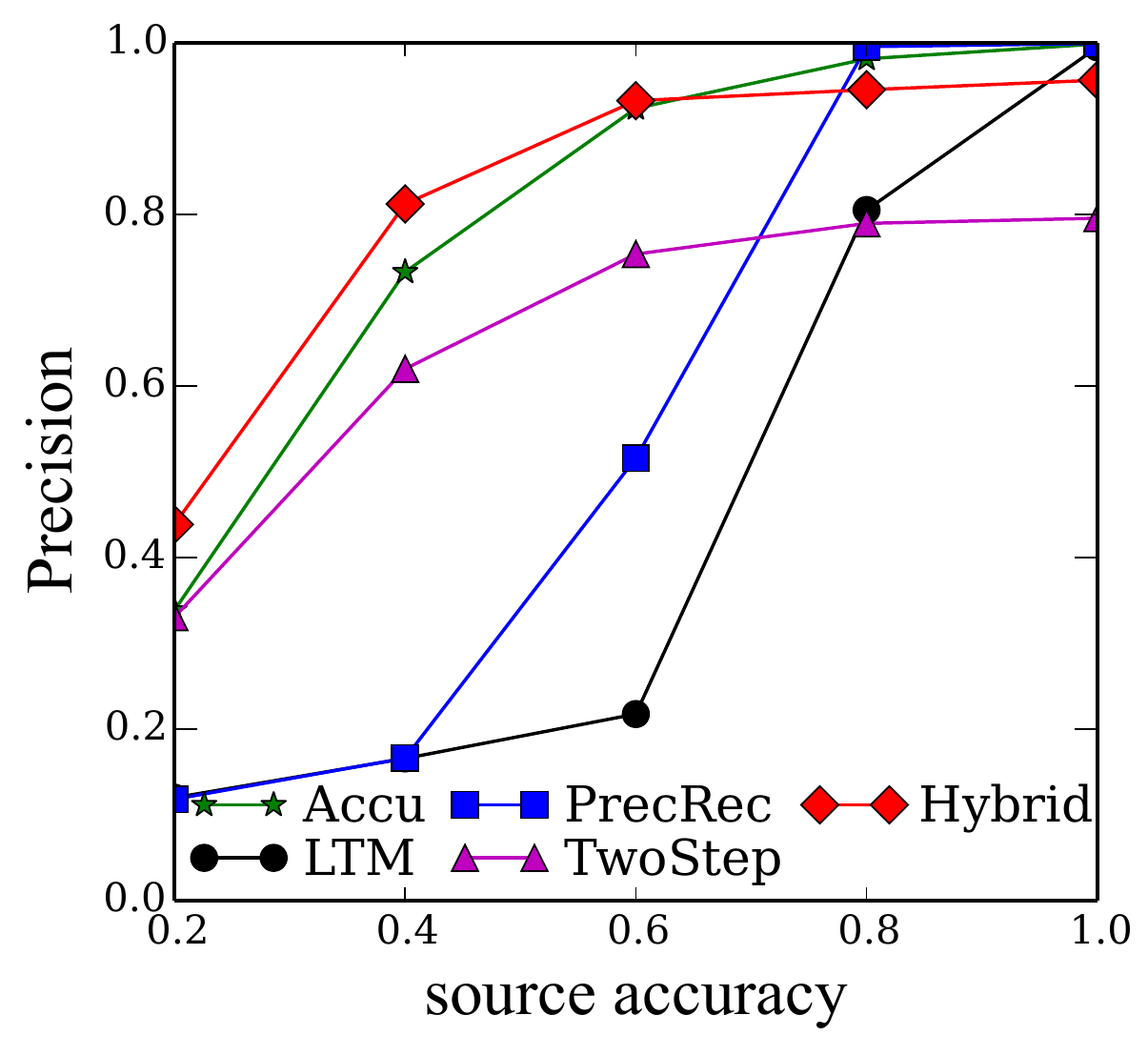} \quad\quad
		\includegraphics[width=.267\textwidth]{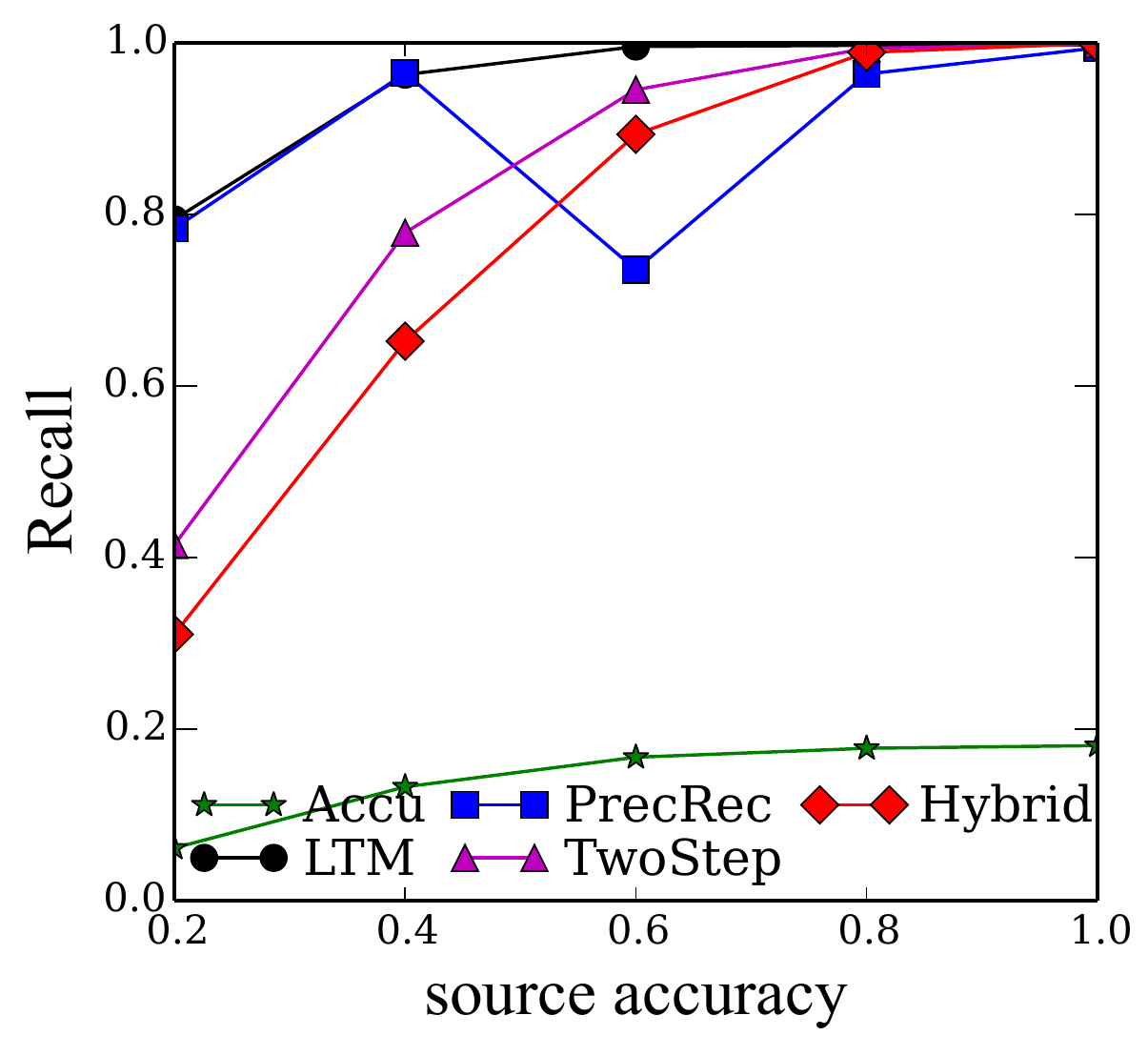}
		\caption{Varying source accuracy. {\sc Hybrid} obtains a significantly higher precision when source accuracy is low.}
		\label{multi:fig:syn_vary_accuracy}
\end{figure*}

\begin{figure*}
		\centering
		\includegraphics[width=.267\textwidth]{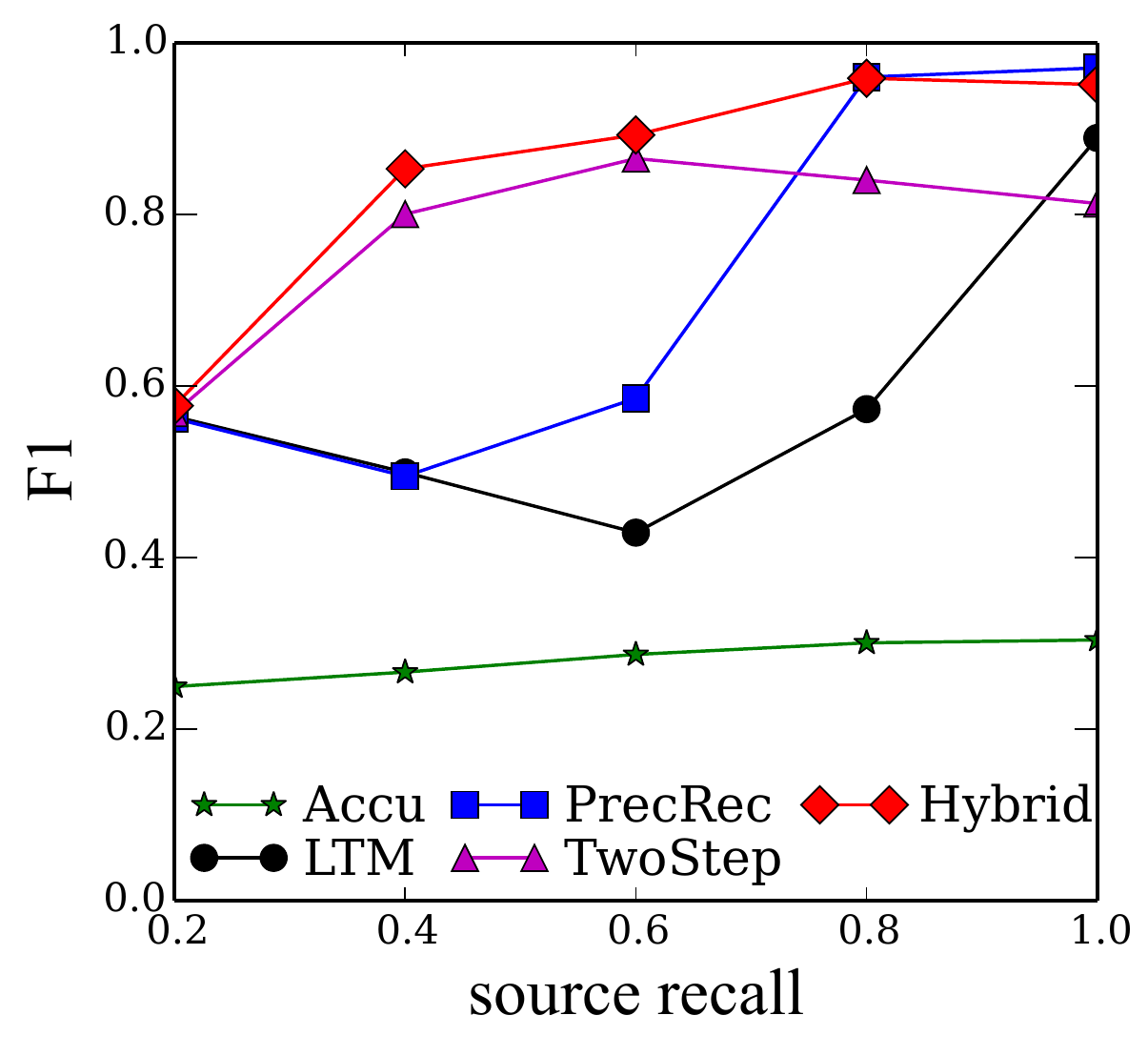} \quad\quad
		\includegraphics[width=.267\textwidth]{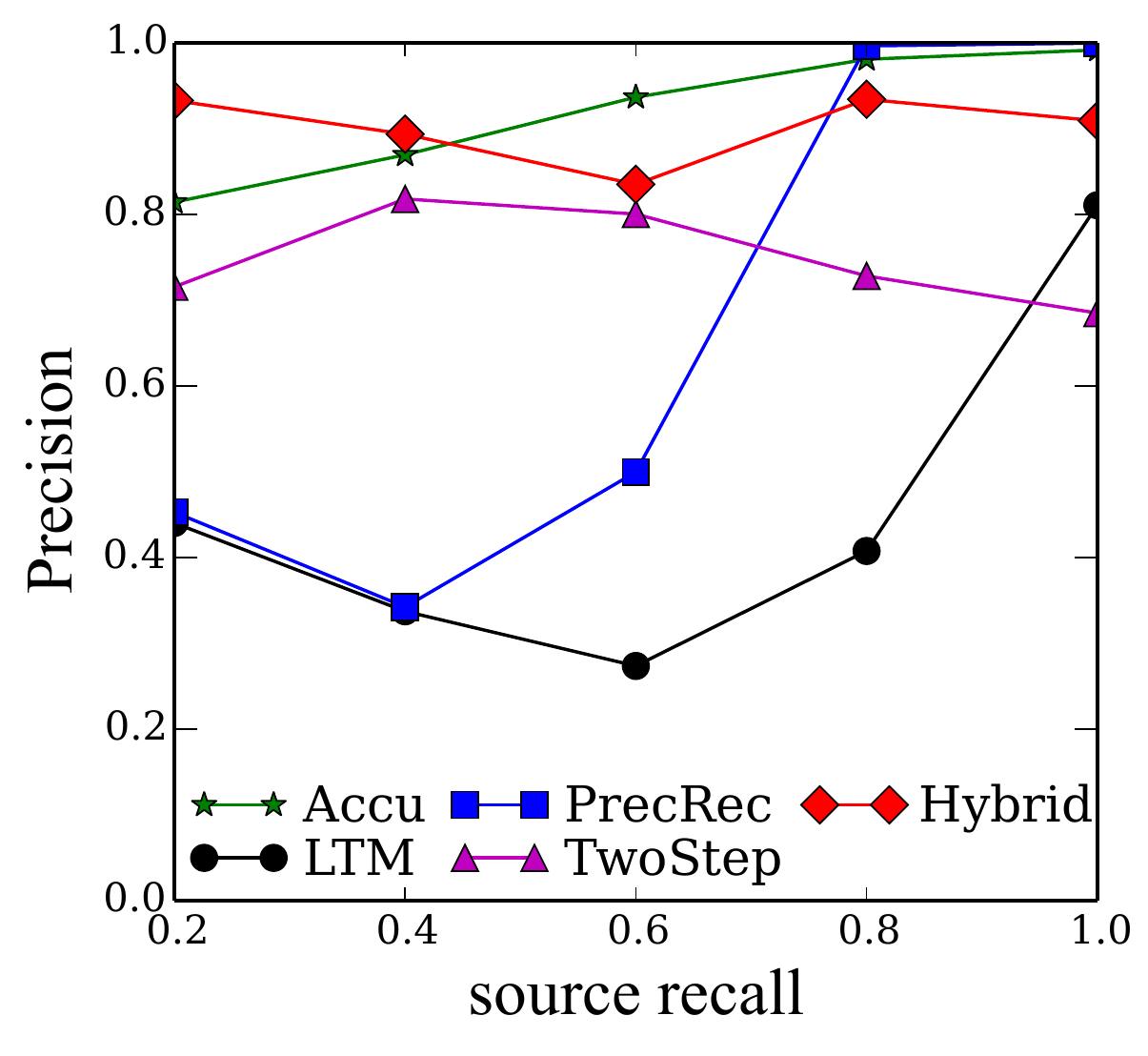} \quad\quad
		\includegraphics[width=.267\textwidth]{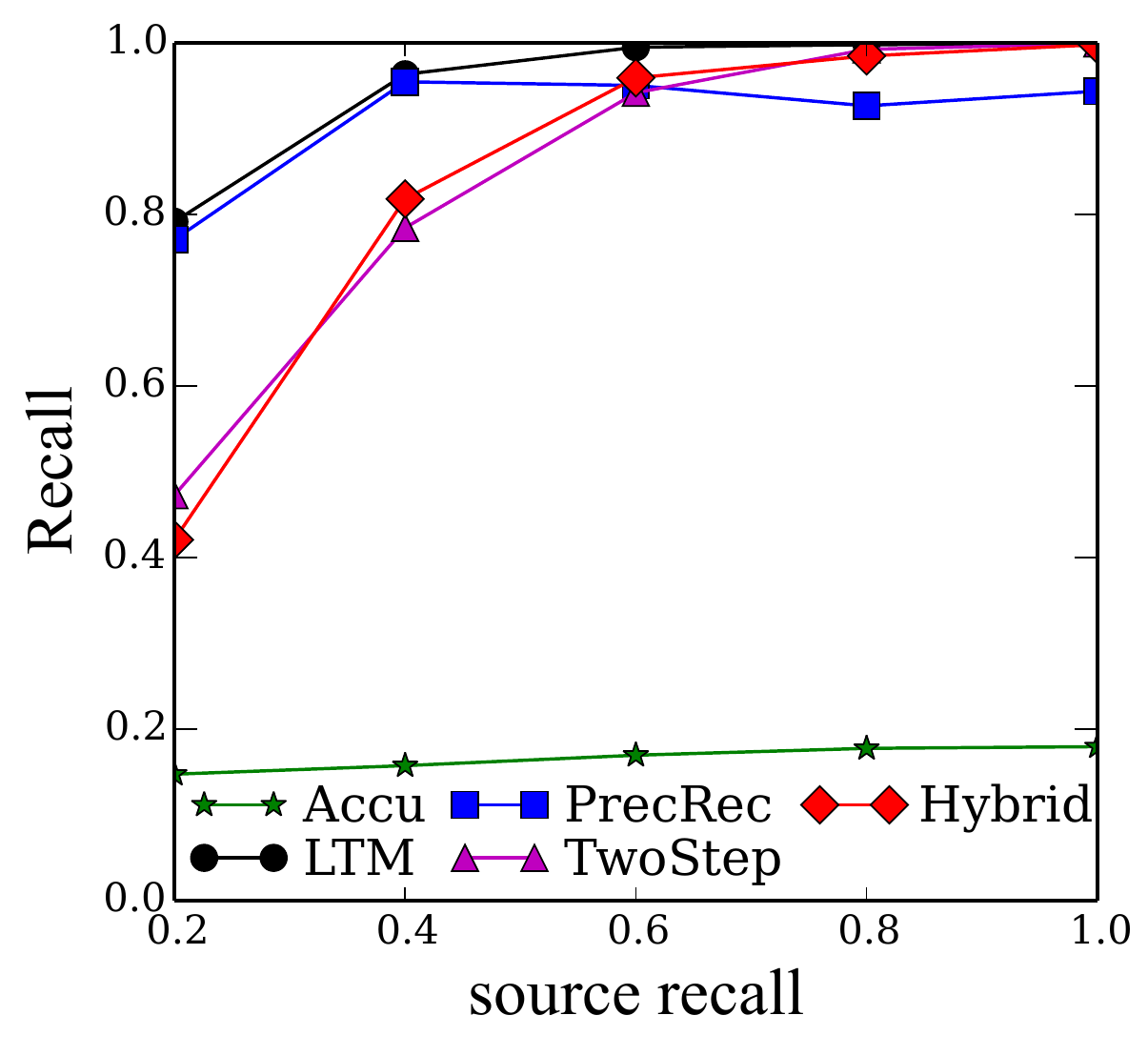}
		\caption{Varying source recall. {\sc Hybrid} gives the highest precision and F1 when the sources have medium recall.}
		\label{multi:fig:syn_vary_capture}
\end{figure*}

\begin{figure*}
		\centering
		\includegraphics[width=.267\textwidth]{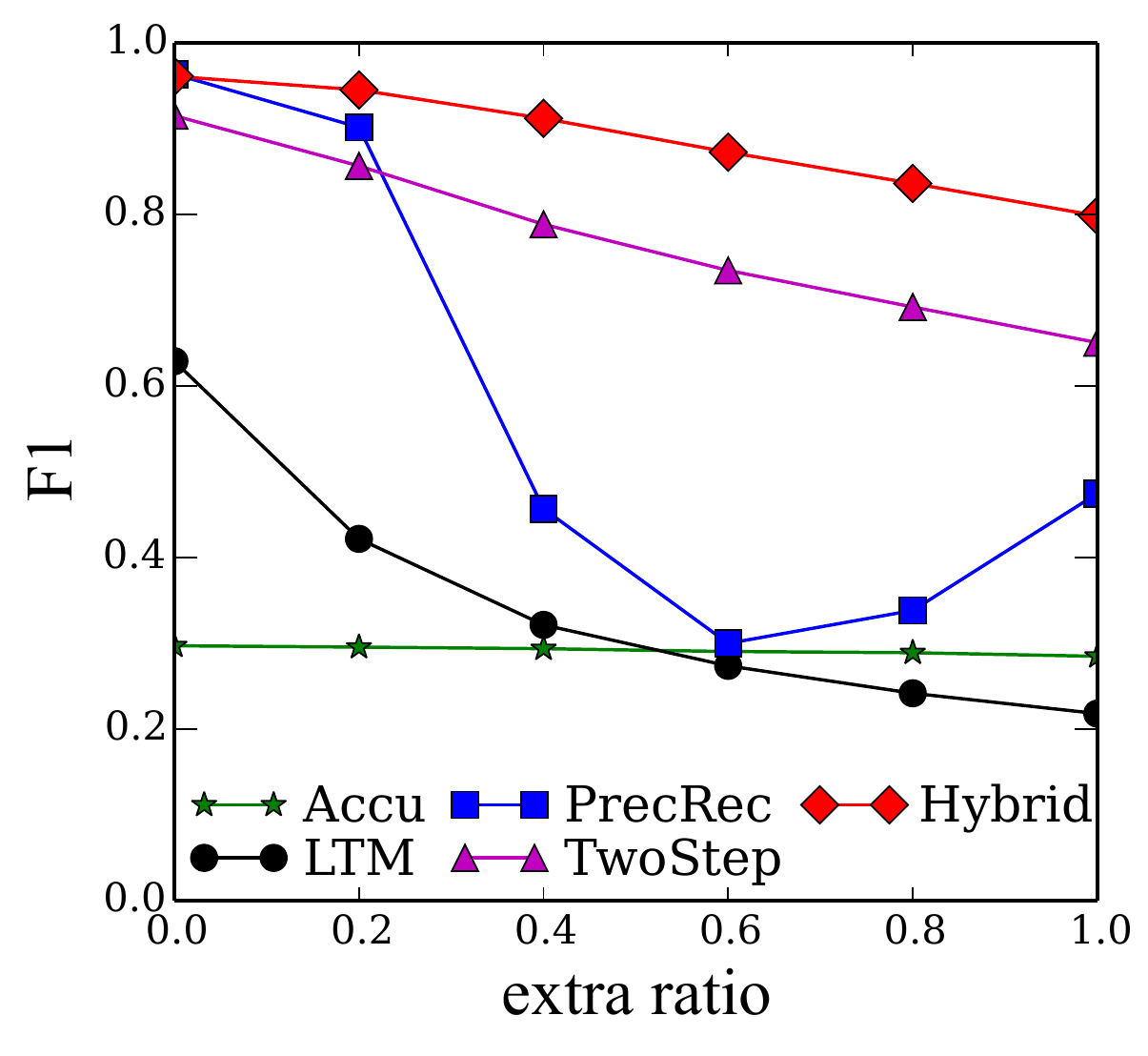} \quad\quad
		\includegraphics[width=.267\textwidth]{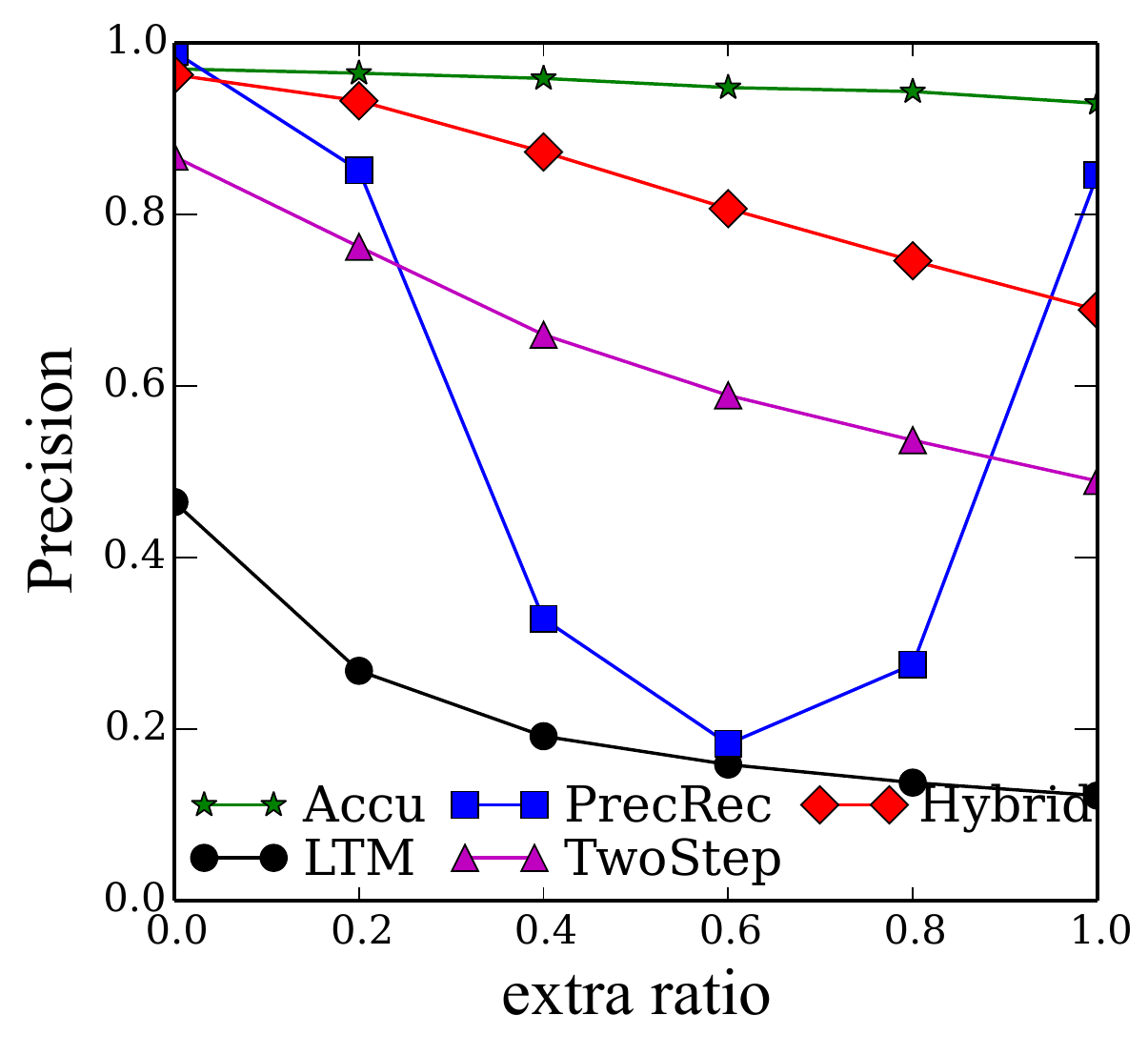} \quad\quad
		\includegraphics[width=.267\textwidth]{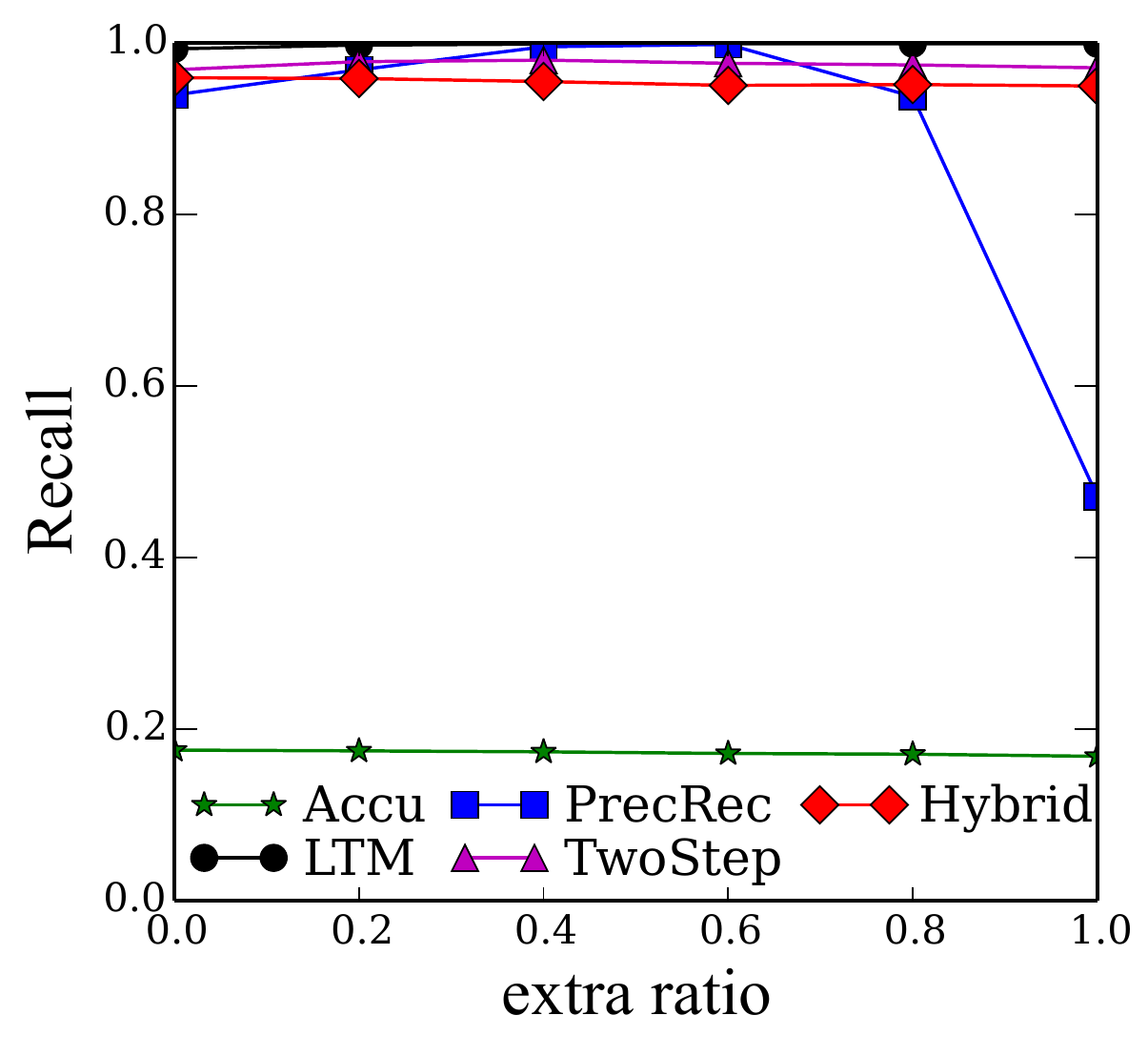}
		\caption{Varying extra ratio. {\sc Hybrid} is the most robust and outperforms the others.}
		\label{multi:fig:syn_vary_extra}
\end{figure*}

Figure~\ref{multi:fig:syn_vary_truthNum} shows the results when we vary the number of truths in data generation.
\model{} can fairly well ``guess'' the number of truths and consistently outperforms the others.
As the number of truths increases, the precision of {\sc Hybrid} remains high, while the precision of {\sc PrecRec} drops.
This is because the {\em extra ratio} is fixed; when there are more truths, there will be more wrong values and {\sc PrecRec} is more sensitive to noise.
Not surprisingly, {\sc Accu} always has the highest precision.
{\sc LTM} and {\sc TwoStep} have low precision: the former lacks a global view of values
for the same data item, while the latter may return false values with weak support. 

Figures~\ref{multi:fig:syn_vary_accuracy}-\ref{multi:fig:syn_vary_extra} plot the results of different methods when we vary source qualities.
As expected, as the quality of the sources drops, the quality of the fusion results drops as well. 
However, we observe that {\sc Hybrid} has the highest F-measure in general and it is the most robust.
It usually gives significantly higher precision than {\sc PrecRec} since it considers the conflicts between provided values as evidence to eliminate wrong values.
While most methods perform better when the source quality increases, {\sc PrecRec} obtains the worst results when the source quality is medium (0.4-0.6). 
This is because, when the sources have similar probability of providing a true value and providing a false value, {\sc PrecRec} is unable to distinguish them.

\section{Related Work}
\label{sec:related}

{\em Data fusion}~\cite{fusionsurvey,fusion-exp} refers to the problem of identifying the truths from different values provided by various sources.
We have provided a high-level review of different approaches in Section~\ref{multi:sec:intro}
and presented a comprehensive experimental study in Section~\ref{multi:sec:exp}.

A model called TEM~\cite{truthExistence} considers in addition whether the truth for a date item exists at all
(\ie, {\em date-of-death} does not exist for an alive person). It is mainly designed for
single-truth scenario.
The method in \cite{confidence} considers the case where most sources provide only
a few triples, thus source quality cannot be reliably estimated.
The {\sc Hybrid} model can be enhanced by these approaches.

\section{Conclusion}
\label{multi:sec:conclusion}

In this paper we present an approach to find true values for an entity from information provided by different sources.
It jointly makes two decisions on an entity: how many truths there are, and what they are. In this way, it allows the existence of multiple truths, while considering the conflicts between different values as important evidence for ruling out wrong values.
Extensive experiments on both real-world and synthetic data show that the proposed outperforms the state-of-the-art techniques, and it is able to obtain a high precision without sacrificing the recall much.

\balance

\bibliographystyle{abbrv} 

\bibliography{fact_check} \normalsize


\begin{appendix}
	
	\section{Proof of Theorem~4.2}
	\label{append:bound}
	
	\indent
	{\sc Theorem}~\ref{multi:the:bound}. {\em 
		Let $d$ be a data item and $n$ be the number of values provided for $d$.
		\begin{itemize}
			\item Algorithm~\ref{multi:algo:greedy} estimates the probability of each provided value 
			in time $O(n^2)$.
			\item For each value $v$ on $d$, we have $|p(v) - \hat{p}(v)| < \frac{1}{6}$,
			where $\hat{p}(v)$ is the exact probability computed by {\sc Hybrid}, 
			and $p(v)$ is the probability obtained by Algorithm~\ref{multi:algo:greedy}. \rbox
		\end{itemize}
	}
	
	\begin{proof}
		We first consider the time complexity of Algorithm~\ref{multi:algo:greedy}.
		Lines~\ref{begin:order}-\ref{end:order} have a complexity of $O(n \log n)$.
		The loops in Lines~\ref{begin:outer} and \ref{begin:valuePr} take $O(n^2)$ time.
		Therefore the overall complexity is $O(n^2)$.
		
		Next we prove the approximation bound of Algorithm~\ref{multi:algo:greedy}. 
		In this proof, we use a tree structure to illustrate the computations made by the full \model{} model; see Figure~\ref{fig:bound_example} as an example.
		The root of the tree represents having not selected any value. 
		A path from the root to a node $v$ represents a possible way of selecting $v$;
		for instance, the path $v_1$-$v_2$-$v_3$ corresponds to the case where we select $v_3$ after selecting $v_1$ and  $v_2$ sequentially (\ie, $\observe = v_1v_2$).
		The children of a node represent candidates for the next truth.
		The number under each node $v$ is the conditional probability $p(v|\observe, \Psi)$.
		By multiplying the numbers along a path, we obtain the {\em probability of the path}. 
		The overall  $p(v|\Psi)$ is thus the sum of the probabilities of all paths ending with $v$. 
		
		\begin{figure}
			\centering
			\includegraphics[width=0.38\textwidth]{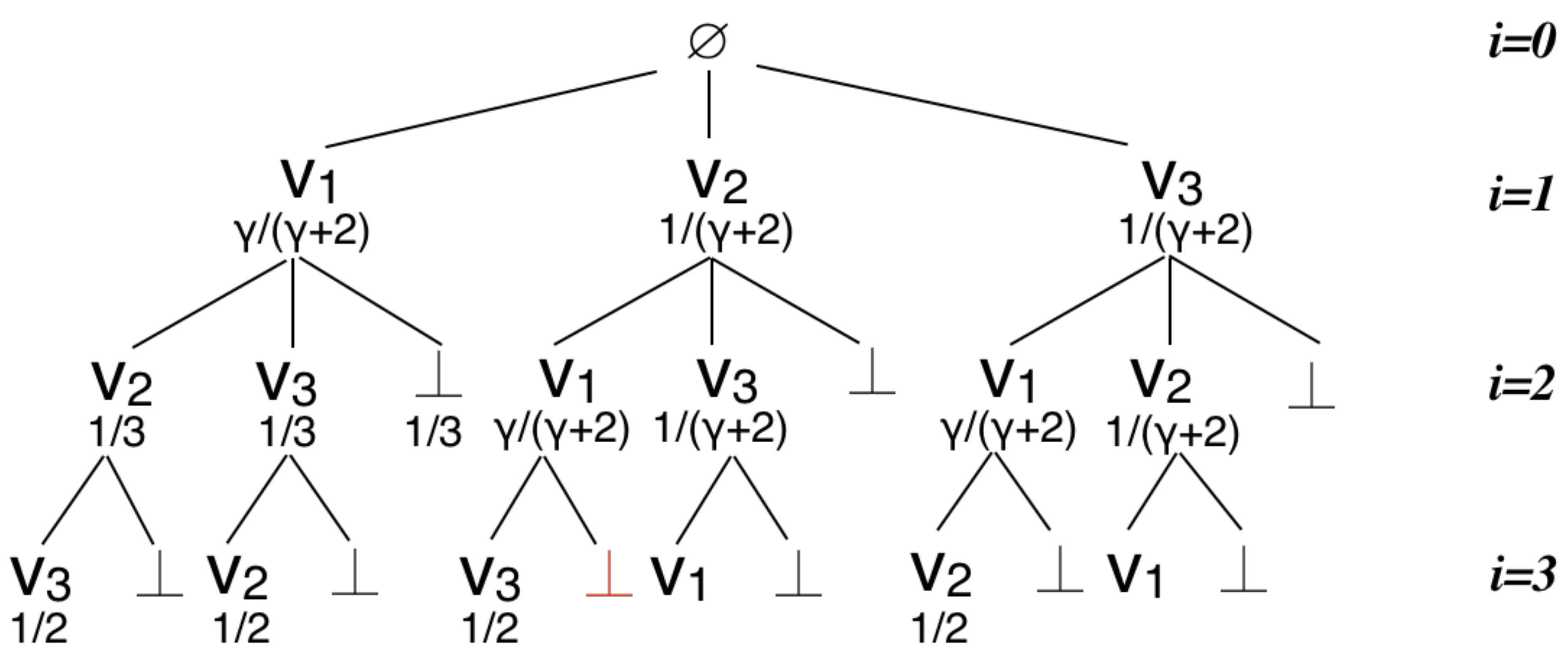}
			\caption{A tree structure to illustrate the full \model{} model.}
			\label{fig:bound_example}
		\end{figure}
		
		Algorithm~\ref{multi:algo:greedy} differs from the full \model{} model in two places:
		(1) it terminates early when $L_i(\perp) > L(v_i)$ without enumerating all possible worlds (Line~\ref{end:terminate}), and (2) it assumes that $v$ has the same conditional probability $p(v|\varGamma_i, \Psi)$ under all possible worlds with size $i-1$ (Line~\ref{end:valuePr}).
		The first one will make the approximated probability $p(v)$ lower than the exact probability $\hat p(v)$ under the full model, while the second one will lead to a higher $p(v)$.
		We next prove by constructing the worst case for each of them.
		
		\smallskip \noindent {\bf Case I:}
		Algorithm~\ref{multi:algo:greedy} terminates early such that $p(v) < \hat p(v)$. \\
		In this case the approximation error is due to the early termination: Algorithm~\ref{multi:algo:greedy} terminates at step $i{-}1$ without increasing $p(v)$ by $p_{j}(v)$ ($j \geq i$) in future steps.
		It is easy to check that Algorithm~\ref{multi:algo:greedy} outputs the same probabilities as the full model if the number of provided values is less than 3.
		So we require at least three values in the domain.
		The earliest possible termination is at step 2 (\ie, $i=3$).
		
		Next we construct a case with three values and Algorithm~\ref{multi:algo:greedy} terminates at step 2, such that the approximation error is reflected in one step (which is $p_3(v)$).
		By definition we have $L(v_3) \leq L(v_2)$ and $L_3(\perp) \geq L_2({\perp})$; to terminate at step 2, we need $L(v_2) < L_2({\perp})$.
		It is easy to see that among the tree values, $v_3$ has the largest probability at step 3.
		To maximize the approximation error cased by $p_3(v_3)$, we need $L(v_3)$ to be maximized and $L_2({\perp})$ as well as $L_3({\perp})$ to be minimized. 
		
		Suppose $L(v_2) = a$, we then have $L(v_3) = L(v_2) = a$; let $\gamma$ be a real number where $\gamma \geq 1$, we have $L(v_1) = \gamma {\cdot}a$.
		Further, let $L_2(\perp) = L_3(\perp) = a + \epsilon$ where $\epsilon$ is a very small constant. As usual, we have $L_1(\perp) = 0$.
		With the above setting, we compute all conditional probabilities following Section~\ref{multi:sec:cond_pr} and illustrate in Figure~\ref{fig:bound_example} (we omit $\epsilon$).
		
		Next we compute the overall probability for $v_3$ using the full \model{} model and Algorithm~\ref{multi:algo:greedy} respectively.
		
		\noindent For the full \model{} model, we find all paths ending with $v_3$ at each level of the tree:
		
		Level 1: $\hat p_1(v_3) = \frac{1}{\gamma+2}$;
		
		Level 2: $\hat p_2(v_3) = \frac{\gamma}{\gamma+2} \cdot \frac{1}{3} + \frac{1}{\gamma+2} \cdot \frac{1}{\gamma+2}$;
		
		Level 3: $\hat p_3(v_3) =  \frac{\gamma}{\gamma+2} \cdot \frac{1}{3} \cdot \frac{1}{2}+ \frac{1}{\gamma+2} \cdot \frac{\gamma}{\gamma+2}\cdot \frac{1}{2}$;
		
		Finally, $\hat p(v_3) = \hat p_1(v_3) + \hat p_2(v_3)  + \hat p_3(v_3) = \frac{1}{\gamma+2} + \frac{\gamma+1}{2(\gamma+2)}$.
		
		\noindent For Algorithm~\ref{multi:algo:greedy}, it terminates after 2 steps:
		
		When $i=1$: $p_1(v_3) = \frac{1}{\gamma+2}$;
		
		When $i=2$: $p_2(v_3) = (1-\frac{1}{\gamma+2}) \times \frac{1}{3} = \frac{\gamma+1}{3(\gamma+2)}$;
		
		Finally, $p(v_3) = p_1(v_3) + p_2(v_3) = \frac{1}{\gamma+2} + \frac{\gamma+1}{3(\gamma+2)}$.
		
		\noindent Therefore
		$\hat p(v_3) -  p(v_3) = \frac{1}{6} \cdot \frac{\gamma+1}{\gamma+2} < \frac{1}{6}$.
		
		\smallskip \noindent {\bf Case II:}
		We assume the same conditional probability under all possible worlds such that $p(v) > \hat p(v)$.
		
		Recall that Eq.~\eqref{multi:eq:upper_v} assumes that in each step $i$, $v$ has the same conditional probability $p(v|\Gamma_i, \Psi)$ under all possible worlds where $v$ is not present yet. 
		This $p(v|\Gamma_i, \Psi)$ is an upper bound of the real conditional probability; we obtain the largest difference between $p(v|\Gamma_i, \Psi)$ and the real conditional probability when the possible worlds end with $\perp$ (the real probability is 0). 
		We thus construct a case where $p_2(\perp)  p(v_3|\Gamma_3, \Psi)$ is maximized, leading to an over-estimation of $p_3(v_3)$.
		
		Similar to the previous case, we assume $L(v_1) = \gamma{\cdot} a$ and $L(v_2) = a$.
		In this case we want Algorithm~\ref{multi:algo:greedy} to continue when $i=3$, so that $p_3(v_3)$ is added to the overall probability of $v_3$.
		Therefore we require $L_2(\perp) \leq L(v_2) $;
		to make $p_2(\perp)$ larger, we have $L_2(\perp) = L(v_2) = a$.
		Given that $L(v_3) \geq L(v_2)$, to maximize $p(v_3|\Gamma_3, \Psi)$, we need $L(v_3) = L(v_2) = a$. 
		
		With the above setting, the probabilities computed by the full model remain the same, but Algorithm~\ref{multi:algo:greedy} continues when $i=3$:
		
		$p_3(v_3) = (1 - \frac{1}{\gamma+2} - (1-\frac{1}{\gamma+2}) \times \frac{1}{3}) \times \frac{1}{2} = \frac{\gamma+1}{3(\gamma+2)}$;
		
		$p(v_3) = p_1(v_3) + p_2(v_3)  + p_3(v_3) = \frac{1}{\gamma+2} + \frac{\gamma+1}{3(\gamma+2)} + \frac{\gamma+1}{3(\gamma+2)} = \frac{1}{\gamma+2} + \frac{2(\gamma+1)}{3(\gamma+2)}$.
		
		Therefore $p(v_3) - \hat p(v_3) = \frac{1}{6} \cdot \frac{\gamma+1}{\gamma+2}< \frac{1}{6}$.
		
		\noindent Combining the two cases, we have $|p(v) - \hat{p}(v)| < \frac{1}{6}$. 
	\end{proof}
	
\end{appendix}

\end{document}